\documentclass[pra,onecolumn,superscriptaddress,aps,5pt]{revtex4}
\usepackage{amsfonts}
\usepackage{amssymb,latexsym,amsmath,amsthm}
\usepackage{graphicx}
\usepackage{dcolumn}
\usepackage{times}
\usepackage{subfigure}
\usepackage[toc,page,title,titletoc,header]{appendix}
\usepackage[usenames]{color}
\usepackage{ulem}

\newcommand{\sech}{\mathrm{sech}}

\newtheorem{lemma}{Lemma}
\newtheorem{proposition}{Proposition}
\newcommand{\ba}{\begin{array}}
\newcommand{\ea}{\end{array}}
\newcommand{\eps}{\epsilon}

\begin{document}

\title{Characterizing Traveling Wave Collisions in Granular  Chains Starting from Integrable Limits: the case of the KdV and the Toda Lattice}

\author{Y. Shen}
\affiliation{Department of Mathematical Sciences, University of Texas at Dallas,
Richardson, TX 75080, USA}

\author{P. G. Kevrekidis}
\affiliation{Department of Mathematics and Statistics, University of Massachusetts,
Amherst, MA 01003-4515, USA}

\author{S. Sen}
\affiliation{Department of Physics, State University of New York, Buffalo, New York 14260-1500, USA}

\author{A. Hoffman}
\affiliation{Franklin W. Olin College of Engineering, Needham, MA 02492}

\begin{abstract}
Our aim in the present work is to develop approximations for the collisional
dynamics of traveling waves in the context of granular chains in the 
presence of  precompression. To that effect, we aim to quantify approximations
of the relevant Hertzian FPU-type lattice through both the Korteweg-de Vries
(KdV) equation and the Toda lattice. Using the availability in such 
settings of both 1-soliton and 2-soliton solutions in explicit analytical
form, we initialize such coherent structures in the granular chain
and observe the proximity of the resulting evolution to the underlying
integrable (KdV or Toda) model. While the KdV offers the possibility
to accurately capture collisions of solitary waves propagating in the
same direction, the Toda lattice enables capturing both co-propagating
and counter-propagating soliton collisions. The error in the approximation
is quantified numerically and connections to bounds established
in the mathematical literature are also given. 
\end{abstract}

\maketitle

\section{Introduction}

Granular crystals are material systems based on the assembly of particles in 
one-, two- and three-dimensions inside a matrix (or a holder) in ordered 
closely packed configurations in which the grains are in contact with each other \cite{Coste1997,Coste1999,Coste2008,Sen2005,Daraio2004,Daraio2006,deBilly2000,DeBilly2006,Gilles2003,Nesterenko1983,Sen1998,Nesterenko2001,Porter2008,Porter2009,Rosas2004,Sen2008,Shukla1991,Molinari2009,Herbold2007}.
The fundamental building blocks constituting such systems are macroscopic 
particles of spherical, toroidal, elliptical or cylindrical shapes \cite{Herbold2007b}, arranged in different geometries. The mechanical, and more 
specifically dynamic, properties of these systems are governed by the stress 
propagation at the contact between neighboring particles.  This confers to the overall system a highly nonlinear response dictated, in the case of particles 
with an elliptical or spherical 
contact, by the discrete Hertzian law of contact interaction \cite{Hertz1881,Johnson1985,Sun2011}. Geometry and/or material anisotropy between particles composing the systems allows for the observation of interesting dynamic 
phenomena deriving from the interplay of discreteness and nonlinearity 
of the problem (i.e. anomalous reflections, breathers, energy trapping and impulse fragmentation) \cite{Sen1998,Nesterenko2001,Herbold2007,Daraio2006b,Hong2005,Hong2002,Job2005,Melo2006,Nesterenko2005,Boechler2010,Job2009,Theocharis2009,chong2013,sen1,sen2,sen3,sen4,vakakis1,vakakis2,vakakis3}. These findings open up a large parameter space for new materials design with unique properties sharply 
departing from classical engineering systems. 

One of the prototypical excitations that have been found to arise in
the granular chains are traveling solitary waves, which have been
extensively studied both in the absence~\cite{Nesterenko2001,Sen1998,Sen2008}
(see also~\cite{pegoenglish,ahnert,stefanov1} for a number of 
recent developments), as well as in 
the presence~\cite{stefanov2} of the so-called precompression. 
The precompression is an external strain a priori imposed on the
ends of the chain, resulting in a displacement of the particles
from their equilibrium position.
As has been detailed in these works, the profile of these traveling
waves is fundamentally different in the former, in comparison to the
latter case. Without precompression, waves exist for any speed, 
featuring a doubly exponential (but not genuinely compact) decay law,
while in the case with precompression, waves are purely supersonic
(i.e., exist for speeds beyond the speed of sound in the medium) and
decay exponentially in space. 

In fact, the FPU type lattices such as the one arising also from the Hertzian
chain in the presence of precompression have been studied extensively 
(see~\cite{bookfpu} and 
 references cited therein for an overview of the history of the FPU 
model).  It is known, both formally~\cite{Kruskal1963} and rigorously~\cite{Schneider1999} (on long but finite time scales) that KdV approximates FPU $\alpha$-type lattices 
for small-amplitude, long-wave, low-energy initial data.  This fact has been used in the mathematical literature to determine the shape~\cite{pego1} and dynamical stability \cite{pego2,pego3,pego4} of solitary waves and even of their interactions~\cite{hoffman}.  We remark that the above referenced remarks in the mathematical literature are valid ``for $\epsilon$ sufficiently small'',
where $\epsilon$ is a parameter characterizing the amplitude and
inverse width, as well as speed of the waves above the medium's
sound speed.  One of the aims of the present work is to 
determine the range of the parameter $\epsilon$ 
for which this theory can be numerically validated, an observation that,
in turn, would be of considerable use
to ongoing granular chain experiments.

It is that general vein of connecting the non-integrable
traveling solitary wave interactions of the granular chain
(that can be monitored experimentally) with the underlying 
integrable (and hence analytically tractable) approximations,
that the present work will be following. In particular, 
our aim is to quantify approximations of the Hertzian contact
model to two other models, one continuum and one discrete in which
soliton and multi-soliton solutions are analytically available.
These are, respectively, the KdV equation and the Toda lattice.
The former possesses only uni-directional waves.  Since Hamiltonian lattices are time-reversible, a single KdV equation cannot capture the evolution of general initial data.  It is typical to use a pair of uncoupled KdV equations, one moving rightward and one moving leftward to capture the evolution of general 
initial data~\cite{Schneider1999}.  On
the other hand, the Toda lattice has several benefits as an
approximation of the granular problem. Firstly,
it is inherently discrete, hence it is not necessary to 
use a long wavelength type approximation that is relevant
for the applicability of the KdV reduction~\cite{pego1,dcds}.
Secondly, the Toda lattice admits two-way wave propagation, hence a single equation can capture the evolution of all (small amplitude) initial data.  

Once these approximations are established,
we will ``translate'' two-soliton solutions, as well as superpositions 
of 1-soliton solutions of the integrable models 
into initial conditions of the granular
lattice and will dynamically evolve and monitor their interactions in comparison
to what the analytically tractable approximations (KdV and Toda)
yield for these interactions. We will explore how the error
in the approximations grows, as a function of the amplitude of
the interacting waves, so as to appreciate the parametric regime
where these approximations can be deemed suitable for understanding
the inter-soliton interaction. We believe that such findings will
be of value to theorists and experimentalists alike. On the 
mathematical/theoretical side, they are relevant for appreciating
the limits of applicability of the theory and the sharpness
of its error bounds. On the experimental side, these explicit
analytical expressions provide a yardstick for quantifying 
solitary wave collisions (at least within an appropriate regime)
in connection to the well-characterized by now 
direct observations~\footnote{It is relevant to mention here
that recent developments have enabled a quantitative characterization of
the full displacement and velocity field~\cite{jkyang} and hence offer a ground
for concrete comparisons between experiments and theory for the
waves and their collisions in the setups considered herein.}. 

Our presentation will be structured as follows. In section II, we will
present the analysis and comparisons for the KdV reduction. In section III,
we will do the same for the Toda lattice, examining in this case
both co-propagating
and counter-propagating soliton collisions. Finally, in section IV, we will
summarize our findings and present some conclusions, as well as some 
directions for future study. In the Appendix, we will present some
rigorous technical aspects of the approximation of the FPU solution
by the Toda lattice one.

\section{Connecting the Granular Chain and its Soliton Collisions to the KdV}

Our starting point here will be an adimensional, rescaled form of
the granular lattice problem, with precompression 
$\delta_0$~\cite{Nesterenko2001,stefanov2} that reads:
\begin{equation}
\ddot{y}_n = [\delta_0+y_{n-1}-y_n]^p_+ - [\delta_0+y_{n}-y_{n+1}]^p_+
\end{equation}
where $y_n$ is the displacement of the $n$-th particle from equilibrium, and 
$[x]_+ = max\{0, x\}$. Defining the strain variables as
$u_n = y_{n-1}-y_n$, we obtain the symmetrized strain equation:
\begin{equation}
\ddot{u}_n = [\delta_0+u_{n-1}]^p_+ - 2[\delta_0+u_{n}]^p_+ + [\delta_0+u_{n+1}]^p_+.
\label{eqn2}
\end{equation}
In the context of the KdV approximation~\cite{pego1,pego2,pego3,pego4} 
(see also more recently and more specifically to the granular 
problem~\cite{dcds}), we seek traveling waves at the long wavelength
limit, which is suitable for the consideration of a continuum limit.
We thus use the following spatial and temporal scales 
$X = \epsilon n$, $T = \epsilon\delta_0^{\frac{p-1}{2}} t$. Assuming then
a strain pattern depending on these scales 
$u_n(t) = A(X,T)$, we get
\begin{equation}
\delta_0^{p-1}\partial_T^2 A = \partial_X^2[(\delta_0+A)^p] +\frac{\epsilon^2}{12}\partial_X^4[(\delta_0+A)^p]+\frac{\epsilon^4}{360}\partial_X^6[(\delta_0+A)^p]+\cdots,
\end{equation}
while by consideration of the variable $B = \frac{A}{\delta_0}<1$
measuring the strain as a fraction of the precompression, we can
also use the expansion of the nonlinear term as:
\begin{equation}
(\delta_0+A)^p = \delta_0^p(1+B)^p = \delta_0^p[1+ pB +\frac{1}{2}p(p-1)B^2+\cdots].
\end{equation}
This finally yields:
\begin{equation}
\partial_T^2 B = \partial_X^2[pB +\frac{1}{2}p(p-1)B^2+\cdots] +\frac{\epsilon^2}{12}\partial_X^4[pB +\frac{1}{2}p(p-1)B^2+\cdots]+\frac{\epsilon^4}{360}\partial_X^6[pB +\frac{1}{2}p(p-1)B^2+\cdots]+\cdots.
\end{equation}
Now consider $B(X,T) = B(\xi,\tau)$, with $\xi = X-cT$, $c = \sqrt{p}$, $\tau = \alpha c T$, with $\alpha$ a small parameter, we get 
\begin{equation}
0 = (2\alpha\partial_\xi\partial_\tau-\alpha^2\partial_\tau^2)B + \partial_\xi^2[\frac{1}{2}(p-1)B^2+\cdots] +\frac{\epsilon^2}{12}\partial_\xi^4[B +\frac{1}{2}(p-1)B^2+\cdots]+\frac{\epsilon^4}{360}\partial_\xi^6[B +\frac{1}{2}(p-1)B^2+\cdots]+\cdots.
\end{equation}
We now proceed to drop lower order terms such as 
$O(\alpha^2B)$, $O(\epsilon^2 B^2)$, $O(B^3)$, $O(\epsilon^4B)$, and thus
obtain the KdV approximation of the form:
\begin{equation}
2\alpha\partial_\tau B + \frac{1}{2}(p-1)\partial_\xi(B^2) +\frac{\epsilon^2}{12}\partial_\xi^3 B=0. \label{KdV0}
\end{equation}
Eq. (\ref{KdV0}) after the transformations $\tilde{\tau} = 2^{-\frac{6}{5}}3^{-\frac{2}{5}}(p-1)^{\frac{3}{5}}\alpha^{-1}\epsilon^{-\frac{2}{5}}\tau$, $B = 2^{\frac{1}{5}}3^{\frac{2}{5}}(p-1)^{-\frac{3}{5}}\epsilon^{\frac{2}{5}}\tilde{B}$, $\tilde{\xi} = 2^{\frac{3}{5}}3^{\frac{1}{5}}(p-1)^{\frac{1}{5}}\epsilon^{-\frac{4}{5}}\xi$, can
be converted to the standard form: 
\begin{equation}
\tilde{ B }_{\tilde\tau}+ 3\partial_{\tilde{\xi}}(\tilde{B}^2) +\partial_{\tilde{\xi}}^3 \tilde{B}=0 \label{KdV}
\end{equation}
which has one soliton solutions as:
\begin{equation}
\tilde{B} = 2k^2 \sech^2[k(\tilde{\xi}-4k^2\tilde{\tau})],
\end{equation}
as well as  two soliton solutions given by:
\begin{equation}
\tilde{B} = 8\frac{k_1^2 f_1+k_2^2 f_2 + 2 (k_2-k_1)^2 f_1 f_2 + m (k_2^2 f_1^2 f_2 + k_1^2 f_1 f_2^2 )}{(1+f_1 + f_2 +m f_1 f_2)^2}.
\end{equation}
Here, $f_i = e^{2k_i(4 k_i^2 \tilde \tau - \tilde \xi +s_i)}$, and $m = [(k_2-k_1)/(k_2+k_1)]^2$; see e.g.~\cite{hir,wahl}, as well as the more recent
work of \cite{Marchant}, for more details on multi-soliton solutions
of the KdV. If the initial positions of the two solitons 
satisfy $s_1<s_2$, we need $k_1>k_2$ for the two solitons to collide. 

A typical example of the approximation of collisional dynamics of the solitons
in the granular chain through the KdV is shown in 
Figs.~\ref{twosoliton} and~\ref{twosoliton_image}. The first figure shows
select snapshots of the profile of the two waves in the strain variable
$u_n$ presenting the comparison of the analytical KdV approximation shown
as a dashed (blue) line with the actual numerical granular chain 
evolution,
of Eq.~(\ref{eqn2}) [shown by solid (red) line]. 
In this, as well as in all the cases that follow, we use the
rescalings developed above (and also for the Toda lattice below)
to transform the integrable model solution into an approximate solution
for the granular chain and initialize in our granular crystal numerics
that solution at $t=0$. I.e., the analytical and numerical results share
the same initial condition and their observed/measured differences are
solely generated by the dynamics.
It is clear that the KdV limit properly captures the
individual propagation of the waves and is proximal not only qualitatively
but even semi-quantitatively to the details of the inter-soliton interaction,
as is illustrated from the middle and especially the bottom panels of the 
figure. Nevertheless, there is a quantative discrepancy in tracking the 
positions of the solitary waves, especially so after the collision.
The second figure shows a space-time plot of the very long scale of the
observed time evolution. It is clear from the latter figure that small
amplitude radiation (linear) waves are present in the actual
granular chain, while such waves are absent in the KdV limit, due to
its integrable, radiationless soliton dynamics. 
In fact, these linear radiation 
waves are also clearly discernible as small amplitude
``blips'' in Fig.~\ref{twosoliton}. We believe that the very long
time scales of the interaction of the waves enable numerous ``collisions''
also with these small amplitude waves thereby apparently reducing the
speed of the larger waves in comparison to their KdV counterparts, 
as is observed in the bottom panels of Fig.~\ref{twosoliton}. In that
light, this is a natural consequence of the non-integrability of
our physical system in comparison to the idealized KdV limit. 
Nevertheless, we believe that the latter offers a very efficient 
means for monitoring the solitary wave collisions even semi-quantitatively.
As a final comment on this comparison, we would like to point out that
because of the very slow (long time) nature of the interaction,
we are monitoring the dynamics in a periodic domain, merely for
computational convenience.

The position shifts of the KdV two-soliton solution after the collision are given by $\frac{1}{k_1}\ln{\frac{k_1+k_2}{k_1-k_2}}$ and $-\frac{1}{k_2}\ln{\frac{k_1+k_2}{k_1-k_2}}$ for the faster and slower solitons respectively \cite{Marchant}. If we use the position shift in KdV to predict the relevant position shifts in 
the granular lattice, for the parameters used in Fig. (\ref{twosoliton}), the shift should be $\frac{1}{k_1}(\ln{\frac{k_1+k_2}{k_1-k_2}})/[2^{\frac{3}{5}}3^{\frac{1}{5}}(p-1)^{\frac{1}{5}}\epsilon^{-\frac{4}{5}}]/\epsilon =7.88$ and  $-\frac{1}{k_2^2}(\ln{\frac{k_1+k_2}{k_1-k_2}})/[2^{\frac{3}{5}}3^{\frac{1}{5}}(p-1)^{\frac{1}{5}}\epsilon^{-\frac{4}{5}}]/\epsilon = -11.15$ for the fast and slow soliton respectively. Numerically, we compare the soliton position with and without the collision,  by tracing the peak of the soliton, and accordingly obtain a 
position shift of $8.2$ for the fast soliton and $-11.2$ 
for the slower soliton, in line with our comments above about
a semi-quantitative agreement between theory and numerics. 

\begin{figure}[htbp]
	\begin{center}
	\includegraphics[width=.45\textwidth]{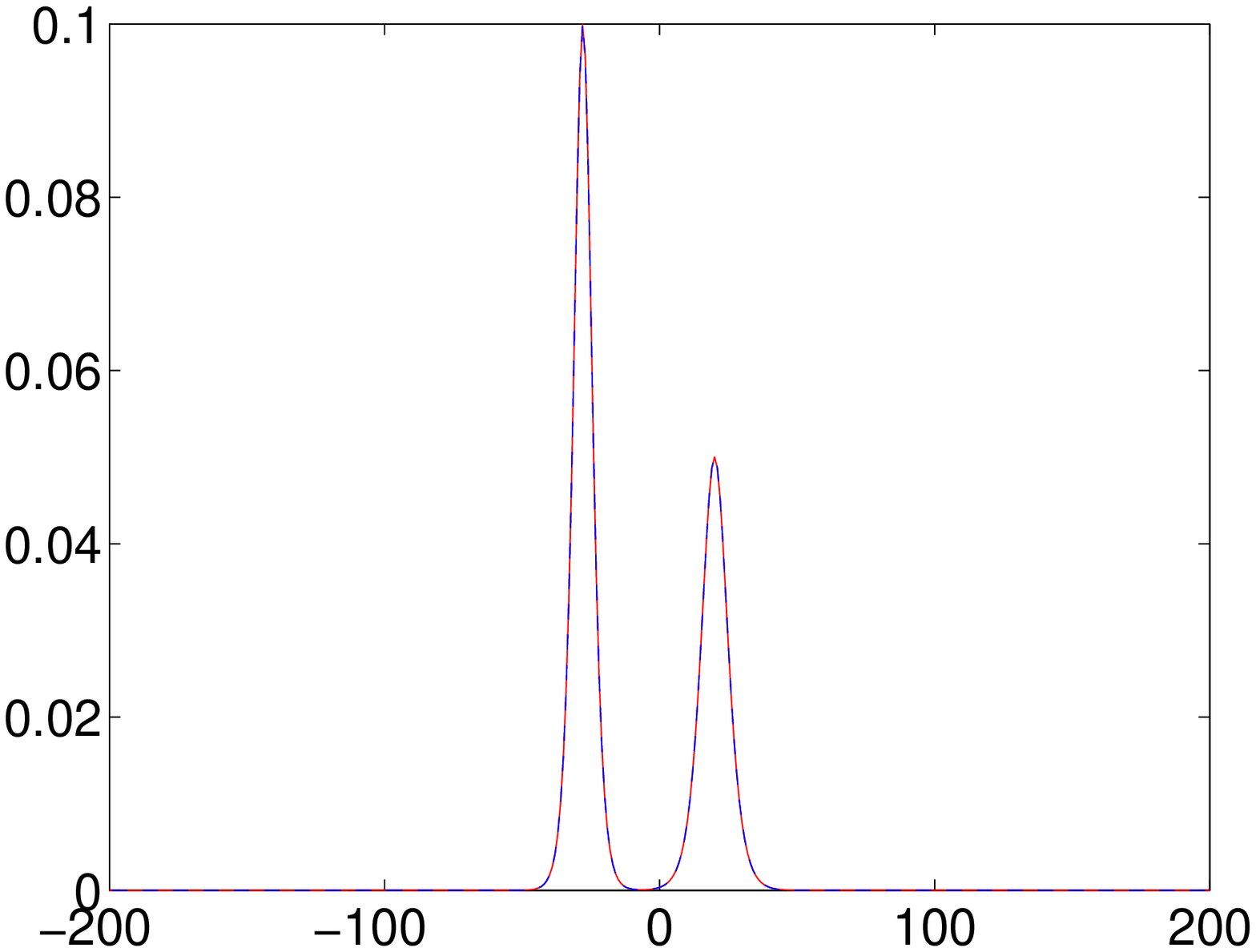}
	\includegraphics[width=.45\textwidth]{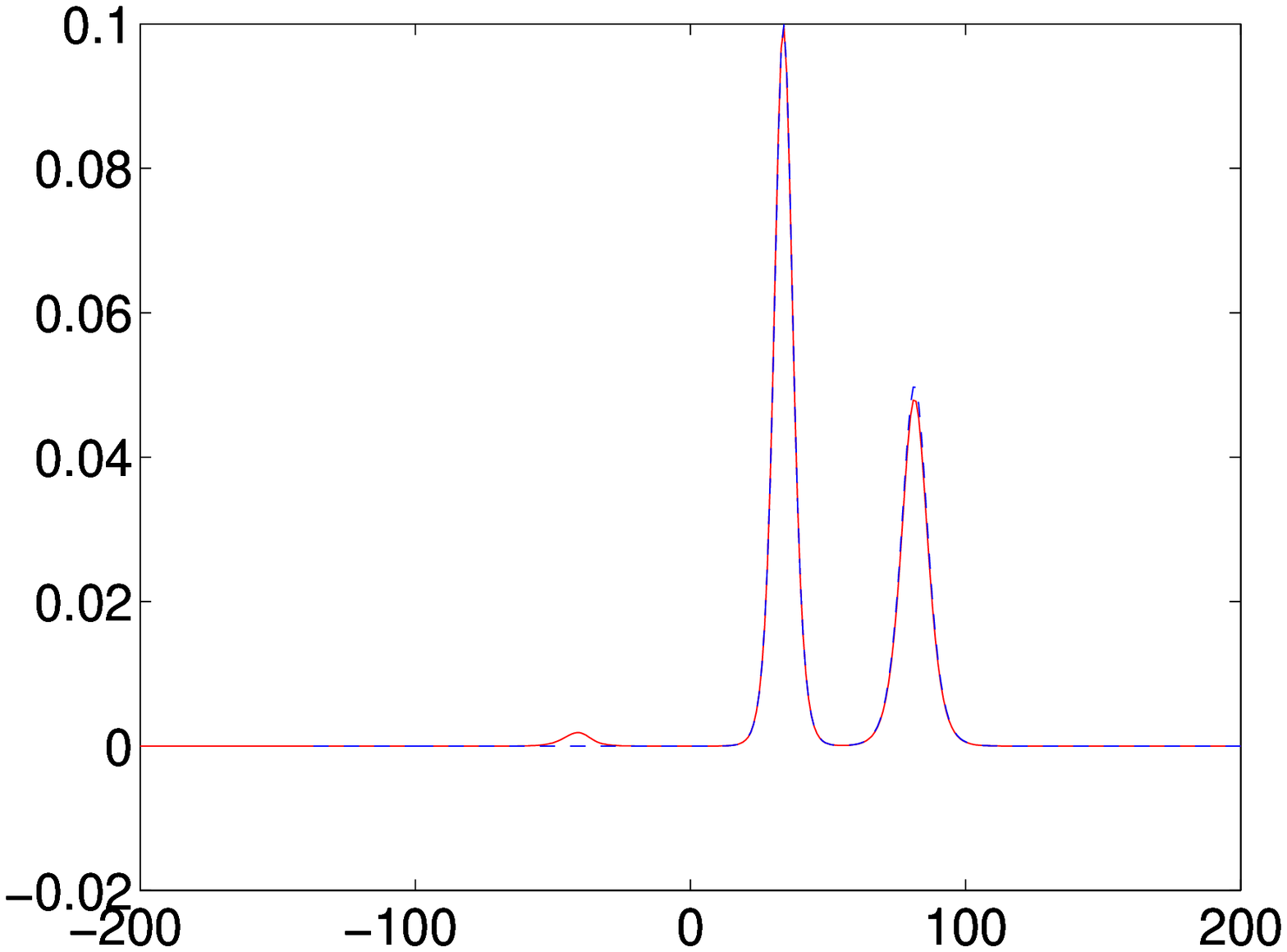}
	\includegraphics[width=.45\textwidth]{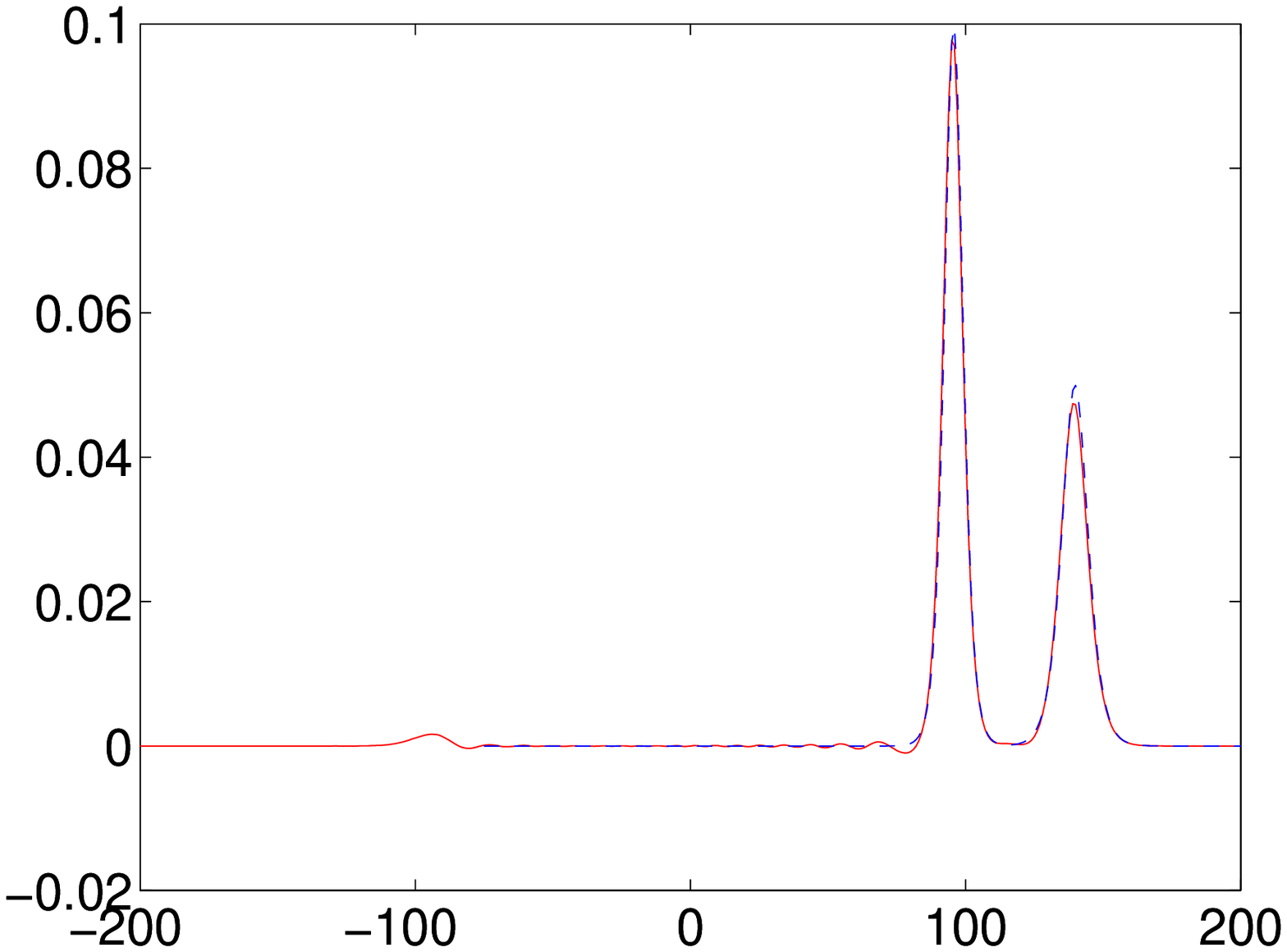}
	\includegraphics[width=.45\textwidth]{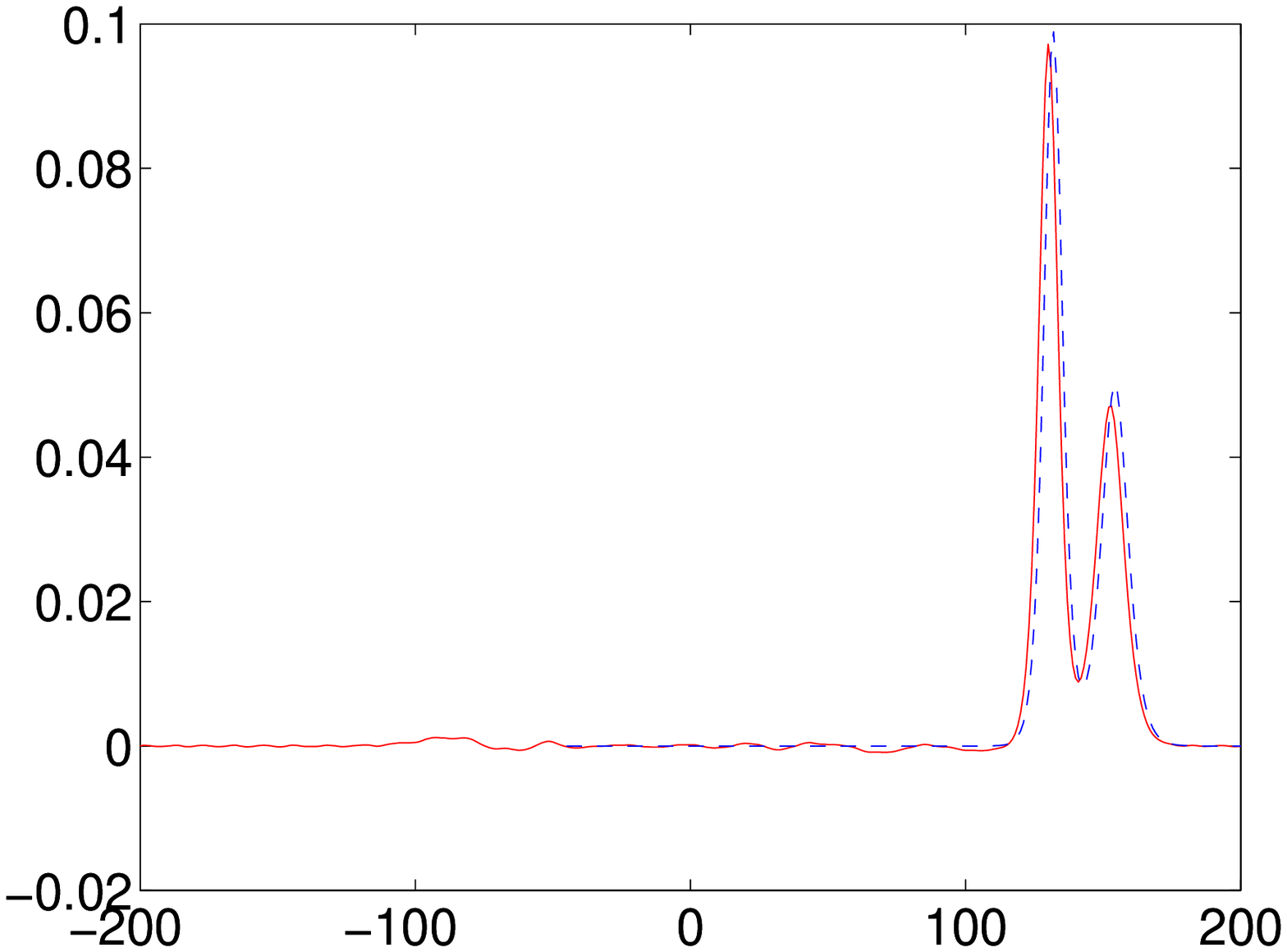}
	\includegraphics[width=.45\textwidth]{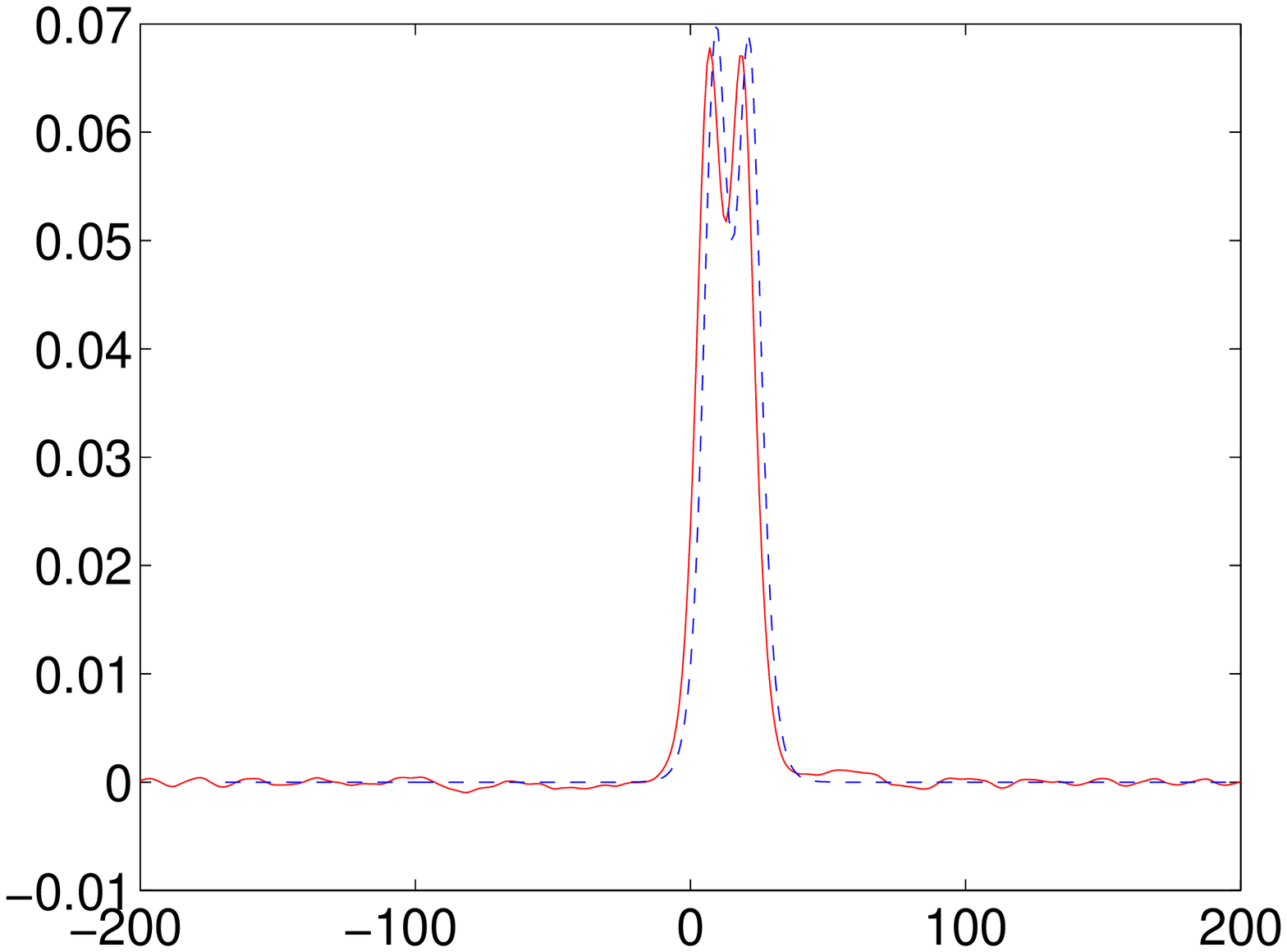}
	\includegraphics[width=.45\textwidth]{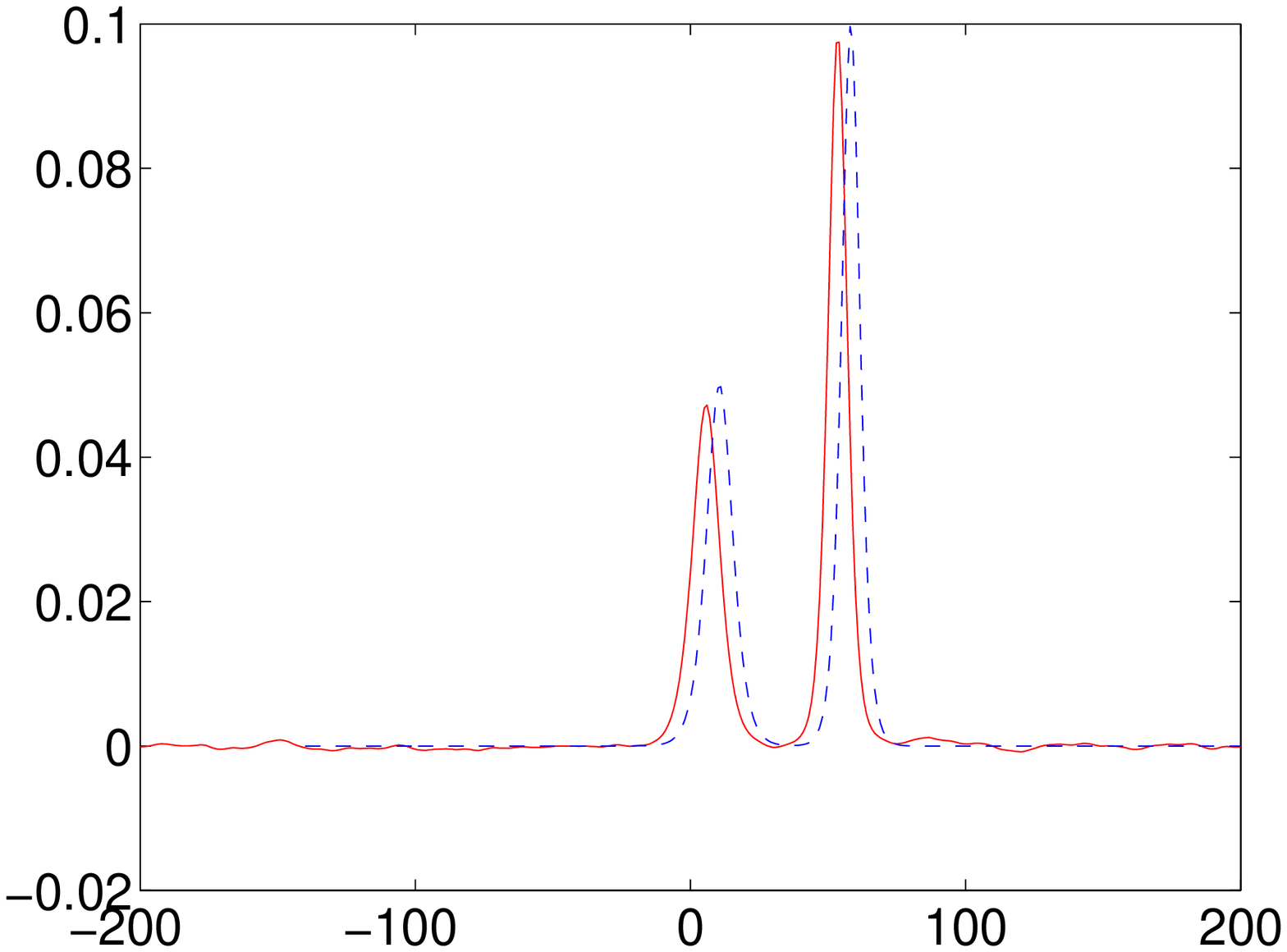}
	\end{center}
	\caption{(Color Online) We consider the comparison of
the collision of two co-propagating solitary waves for the case
of parameters chosen as: 
$p = \frac{3}{2}$, $\delta_0=1$, $\alpha = .1$, $\epsilon = \sqrt{0.1}$.
The initial condition consists of a two-soliton solution 
containing waves of amplitude of $0.1$ and $0.05$ centered at $-20$ and 
$20$, respectively. From top to bottom, left to right snapshots at different
times of the collisional evolution are shown, namely: 
$t=0, 50, 750, 5000, 7500,15000$. The solid (red) line represents
the actual (non-integrable) granular lattice 
numerical evolution dynamics, while the dashed (blue) line 
stems from the qualitatively (and even semi-quantitatively) 
accurate integrable KdV two-soliton approximation.\label{twosoliton} }
\end{figure}

\begin{figure}[htbp]
	\begin{center}
	\includegraphics[width=1\textwidth]{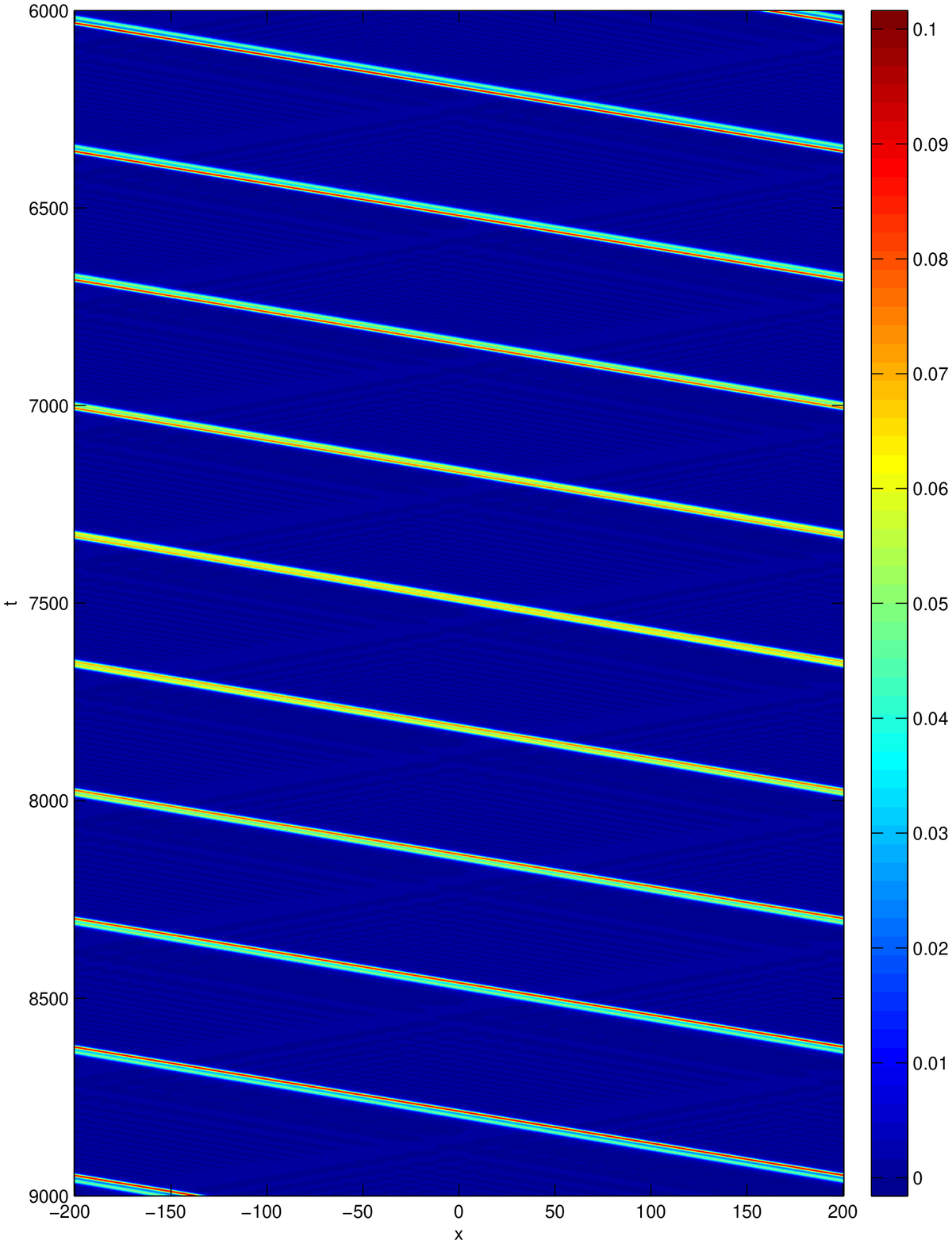}
	\end{center}
	\caption{(Color Online) Here, the parameters and initial data are the 
same as in Fig.~\ref{twosoliton}, but the space-time contour plot of the
strains is shown for the granular lattice evolution.	Notice the long
time scale of the interaction, the exchange of the relative positions of
the solitary waves and the nontrivial presence (and thus impact on the
interaction) of small amplitude linear radiation waves stemming from 
the non-integrability of the model. Periodic boundary conditions have been
employed.}
\label{twosoliton_image}
\end{figure}


\section{Connecting the Granular Chain and its Soliton
Collisions to the Toda Lattice}

The Toda Lattice model has the well-known form~\cite{toda1,toda2,toda3}

\begin{eqnarray}
\ddot{x}_n &=& e^{[x_{n-1}-x_n]}- e^{[x_{n}-x_{n+1}]}\nonumber\\
 &=& [1+(x_{n-1}-x_n)+ \frac{1}{2} (x_{n-1}-x_n)^2 +\frac{1}{6} (x_{n-1}-x_n)^3 +\cdots] -[1+(x_{n}-x_{n+1}) + \frac{1}{2} (x_{n}-x_{n+1})^2 +\frac{1}{6} (x_{n}-x_{n+1})^3 +\cdots]\nonumber\\
&=& (x_{n-1}-2x_n+x_{n+1}) + \frac{1}{2}[(x_{n-1}-x_n)^2-(x_{n}-x_{n+1})^2 ]+\cdots
\label{toda}
\end{eqnarray}
In the 2nd and 3rd lines above, we have expanded the lattice into an 
FPU-$\alpha$ type form (i.e., maintaining the leading order nonlinear
term). On the other hand, a similar expansion (notice that now
no long wavelength assumptions are needed) of our  
granular  chain model reads:
 \begin{eqnarray}
\ddot{y}_n &=& [\delta_0+y_{n-1}-y_n]^p - [\delta_0+y_{n}-y_{n+1}]^p\nonumber\\
&=& \delta_0^p[1+p\frac{y_{n-1}-y_n}{\delta_0}+\frac{1}{2}p(p-1)(\frac{y_{n-1}-y_n}{\delta_0})^2+\cdots]-\delta_0^p[1+p\frac{y_{n}-y_{n+1}}{\delta_0}+\frac{1}{2}p(p-1)(\frac{y_{n}-y_{n+1}}{\delta_0})^2+\cdots]\nonumber\\
&=& p\delta_0^{p-1}\{(y_{n-1}-2y_n+y_{n+1} )+\frac{1}{2}\frac{(p-1)}{\delta_0}[( y_{n-1}- y_n)^2-( y_{n}- y_{n+1})^2 ] +\cdots\}
\label{gl}
\end{eqnarray}

Then, rescaling time and displacements according to 
$\tau = t\sqrt{p\delta_0^{p-1}}$ and 
$ \tilde y_n = \frac{p-1}{\delta_0} y_n$, the relevant Eqn. (\ref{gl})
becomes 
\begin{eqnarray}
\tilde{y}_n'' =(\tilde y_{n-1}-2y_n+\tilde y_{n+1} )+\frac{1}{2}[(\tilde y_{n-1}- \tilde y_n)^2-( \tilde y_{n}-\tilde y_{n+1})^2 ] +\cdots
\label{sgl}
\end{eqnarray}
where $'$ is the derivative with respect to $\tau$. Hence, Eqs. (\ref{sgl}) and Eqn. (\ref{toda}) agree up to second order, and thus the leading
order error in our granular chain approximation by the Toda
lattice will stem from the cubic term (for which it is straightforward
to show that it cannot be matched between the two models i.e., we have
expended all the scaling freedom available within the discrete granular
lattice model).  

To see the closeness of the two models, we define the error term $Y$ by the relation $\tilde{y}_n = x_n + \epsilon Y_n$.  Here $Y$ will remain of order one or smaller and $\epsilon$ controls the size of the error term.  We proceed by using the evolution for $Y$ to control how small we can choose $\epsilon$ while keeping $Y$ of order one over timescales of interest.  We compute  
\begin{eqnarray}
\ddot{Y}_n & = & \eps^{-1}\left(\ddot{\tilde{y}}_n - \ddot{x}_n\right) \nonumber \\
& = & Y_{n+1} + Y_{n-1} - 2Y_n + \epsilon^{-1}\mathrm{Res} + L(x)Y + \eps^{-1}N(\eps Y).
\end{eqnarray}
Here $L(x)$ is a linear operator with a norm that scales roughly like $\|x\|$, $N$ is quadratic and the residual given by the disparity between the interaction potential for Toda and that for the granular chain is:
\[ \ba{lcl} \mathrm{Res} & = & e^{x_{n-1}-x_n}-e^{x_n-x_{n+1}} - \frac{p-1}{p}\left[ 1 + \frac{x_{n-1}-x_n}{p-1}\right]^p + \frac{p-1}{p}\left[ 1 + \frac{x_n - x_{n+1}}{p-1}\right]^p \\ \\
& = & \frac{1}{6}(1-\frac{p-2}{p-1})((x_{n-1}-x_n)^3 - (x_n-x_{n+1})^3) + \mathcal{O}\left((x_{n-1}-x_n)^4 + (x_n-x_{n+1})^4\right). \ea \]

Since the discrete wave equation conserves the $l_2$ norm exactly, the $l_2$ norm of $Y$, for time scales on which $\|L(x)\|T << 1$, will be bounded above by a constant times $T\epsilon^{-1}$ times the $l_2$ norm of $((x_{n-1}-x_n)^3 + (x_n-x_{n+1})^3)$.  In other words $Y$ remains of 
order one on timescale $T$ so long as 
\[ \epsilon << T(\sum_n((x_{n-1}-x_n)^3 - (x_n-x_{n+1})^3)^2)^{(1/2)} \]
In the sequel we will consider solutions for which $(x_n - x_{n+1}) \sim k^2 e^{-kn}$ over timescales $k^{-1}$.  Thus $(x_{n-1}-x_n)^3 - (x_n - x_{n+1})^3 \sim k^7e^{-kn}$ and we obtain an upper bound on the approximation error of $\epsilon << k^{-1}(\sum_n \frac{k^{14}}{1-e^{-k}})^{1/2} \sim k^{5.5}$.  We note that this improves on the estimate of $k^{3.5}$ which appears e.g. in \cite{Schneider1999, hoffman}. [A number of details towards making this argument rigorous are presented
in the Appendix]. After describing the single and multiple solitary 
wave solutions of the Toda lattice, we will return to the numerical 
examination of the validity of this concrete prediction.


In starting our comparison of the evolution of Toda lattice 
analytical solutions with the granular crystal 
dynamical evolution, 
we consider the single soliton solution of the Toda lattice of form 
\begin{eqnarray}
{x}_n =-\ln\left\{\frac{1+\exp[-2kn \pm 2(\sinh k)t ]}{1+\exp[-2k(n-1) \pm 2(\sinh k)t]}\right\}.
\label{TodaOne}
\end{eqnarray}
By composing two counter propagating solitons we get 
a typical dynamical evolution such as the one
presented in Figs.~\ref{Toda1} and~\ref{Toda2}. Once again
(as in the KdV case), the former
represents the snapshots at specific times, while the latter
the contour plot of the strain variable evolution (as will be the
case in all the numerical experiments presented herein). The figure
contains the comparison of 3 waveforms. The solid (red) one is
from the time integration of the granular  chain dynamics. The
dashed (blue) line is a plain superposition of two one-solitons of the
Toda lattice, while the dash-dotted (green) line shows the 
evolution of the Toda lattice. Detailed examination of the
latter two suggests that 
the dashed and the dash-dotted curves do not perfectly coincide
(although such differences are not straightforwardly discernible
in the scale of Fig.~\ref{Toda1}).
This is the well-known feature of the presence of {\it phase shifts}
as a result of the solitonic collisions in the integrable dynamics.
It is however relevant
to add here that admittedly not only qualitatively but even
quantitatively the Toda lattice appears to be capturing the
counter-propagating soliton dynamics of our granular  chain, both
before, during and after the collision.

\begin{figure}[htbp]
	\begin{center}
	\includegraphics[width=.45\textwidth]{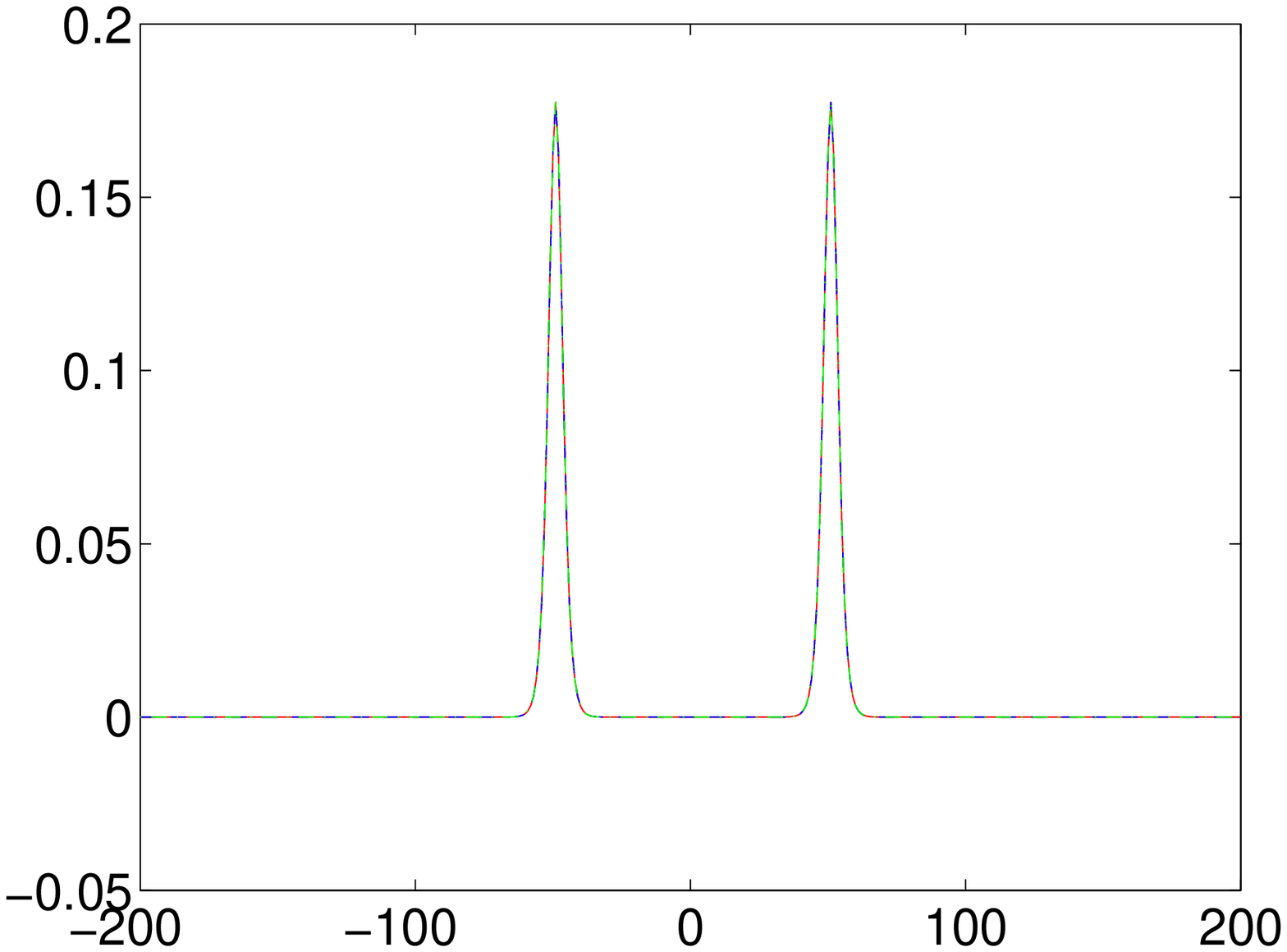}
	\includegraphics[width=.45\textwidth]{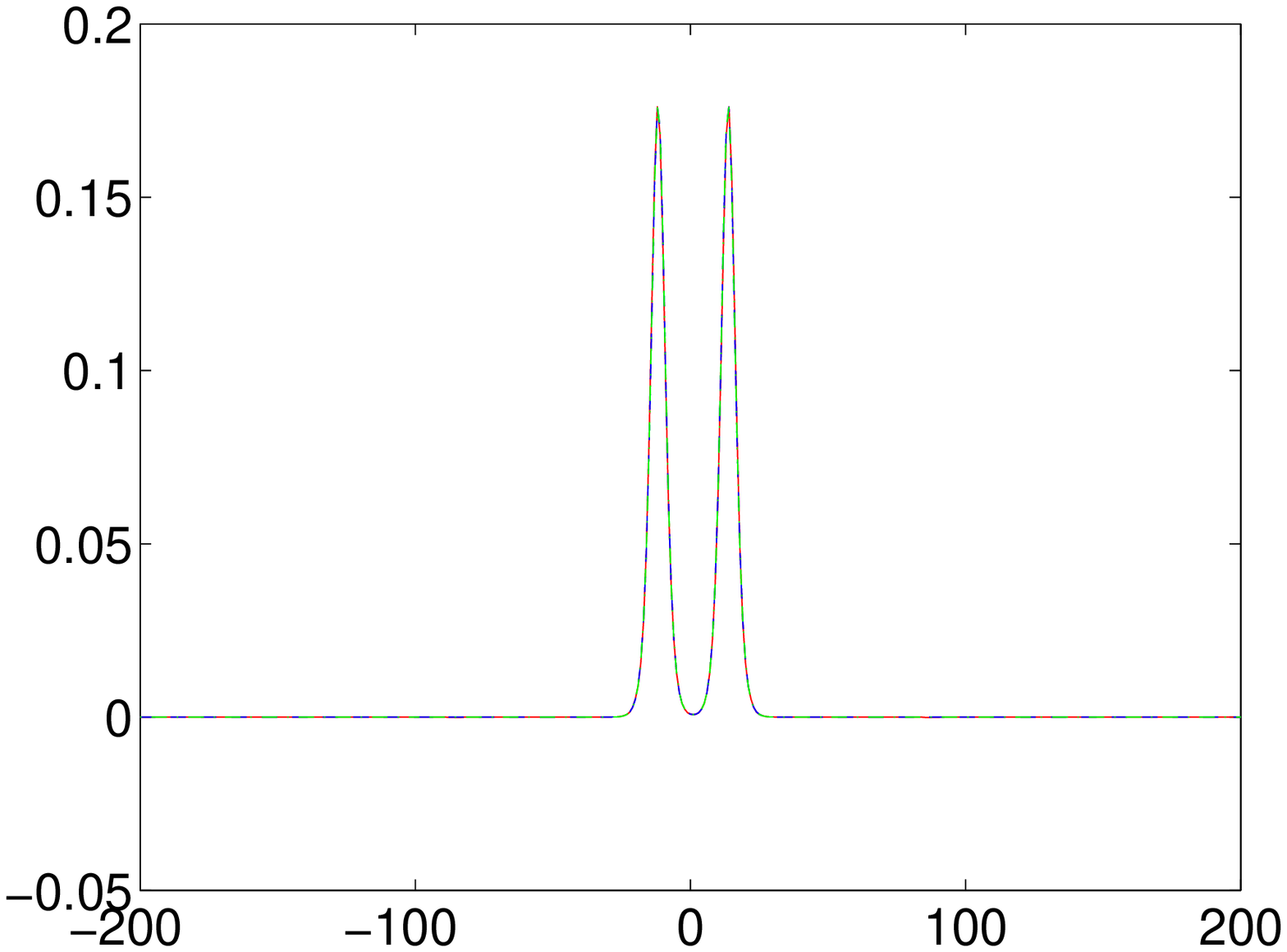}
	\includegraphics[width=.45\textwidth]{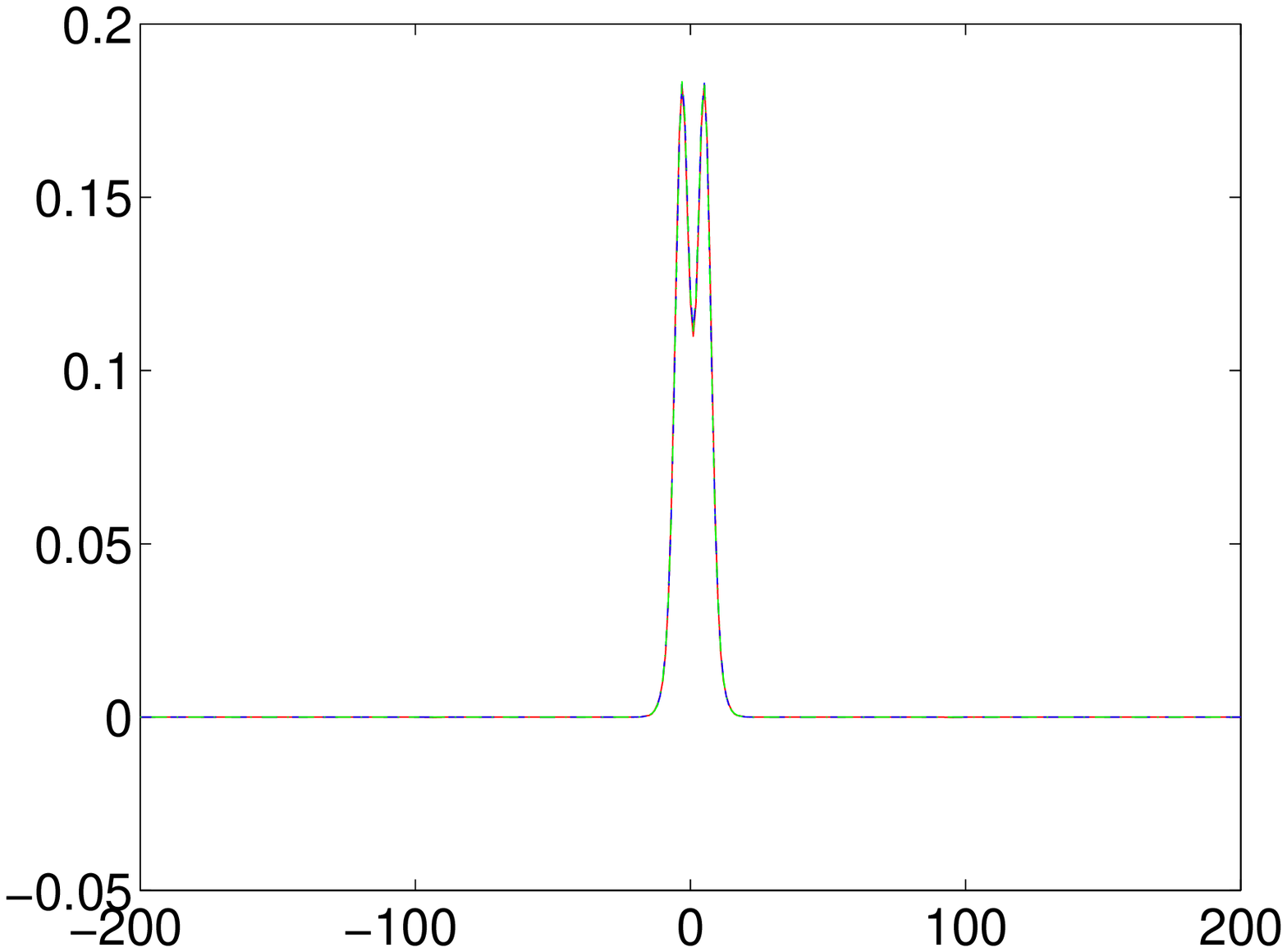}
	\includegraphics[width=.45\textwidth]{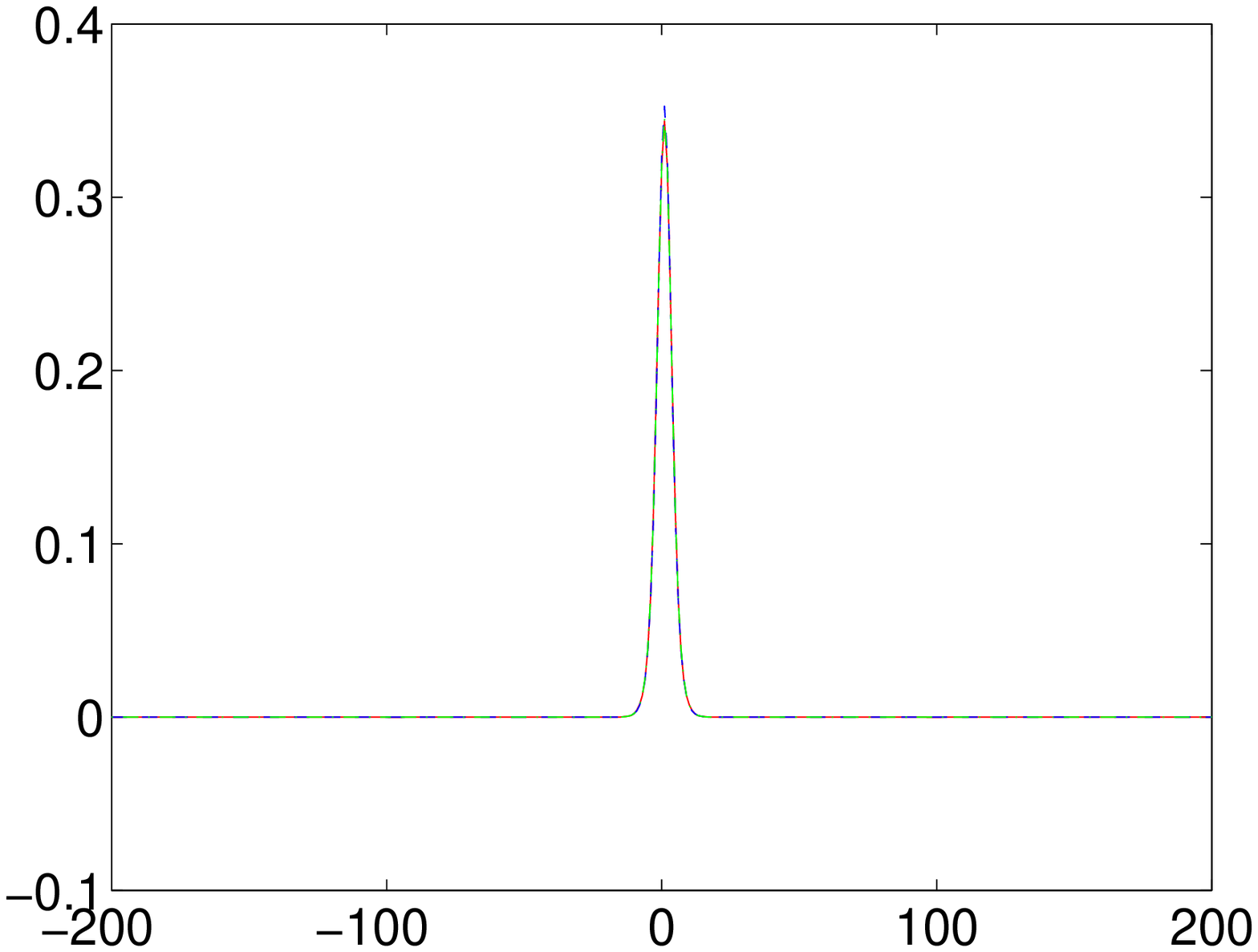}
	\includegraphics[width=.45\textwidth]{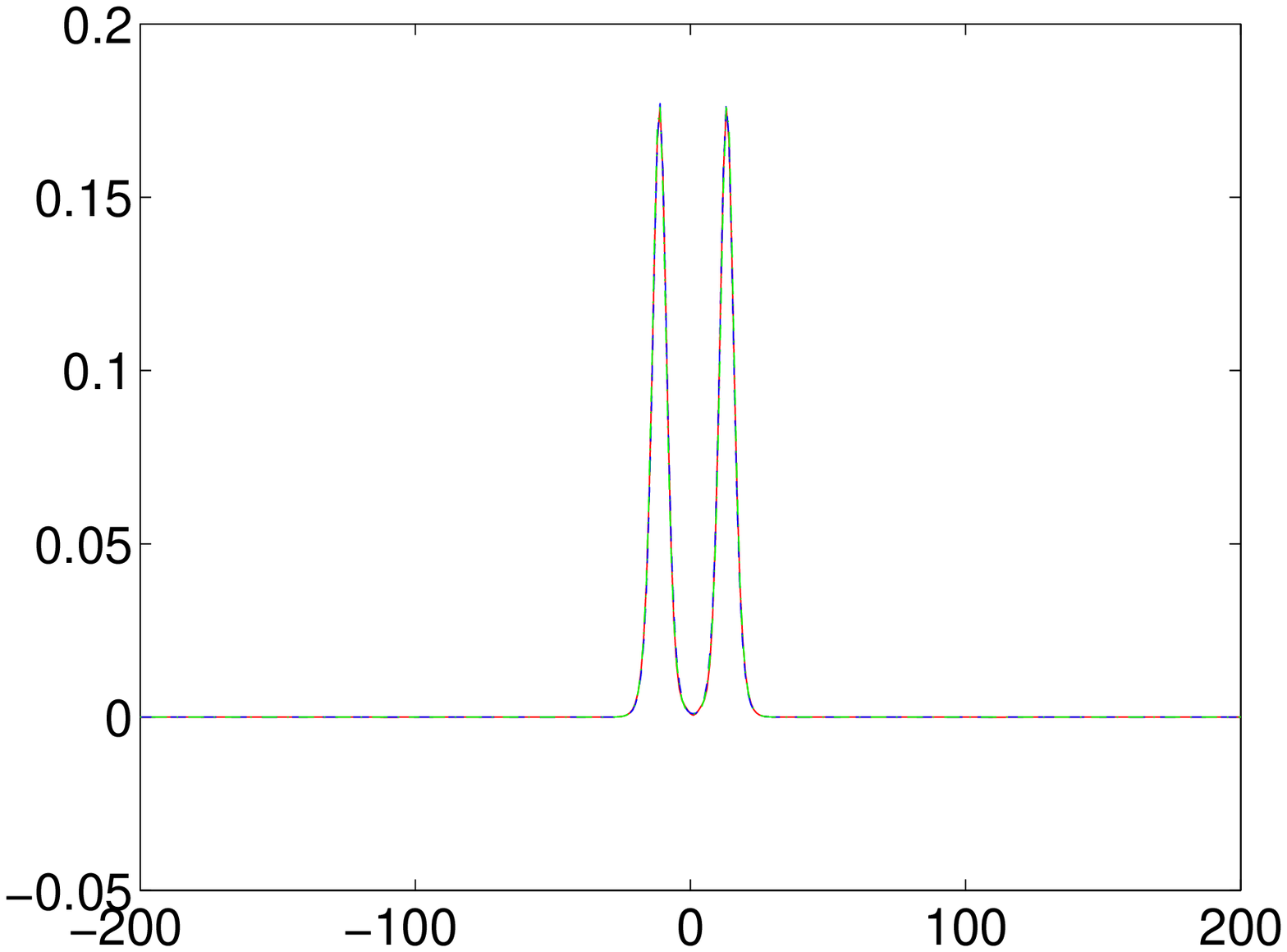}
	\includegraphics[width=.45\textwidth]{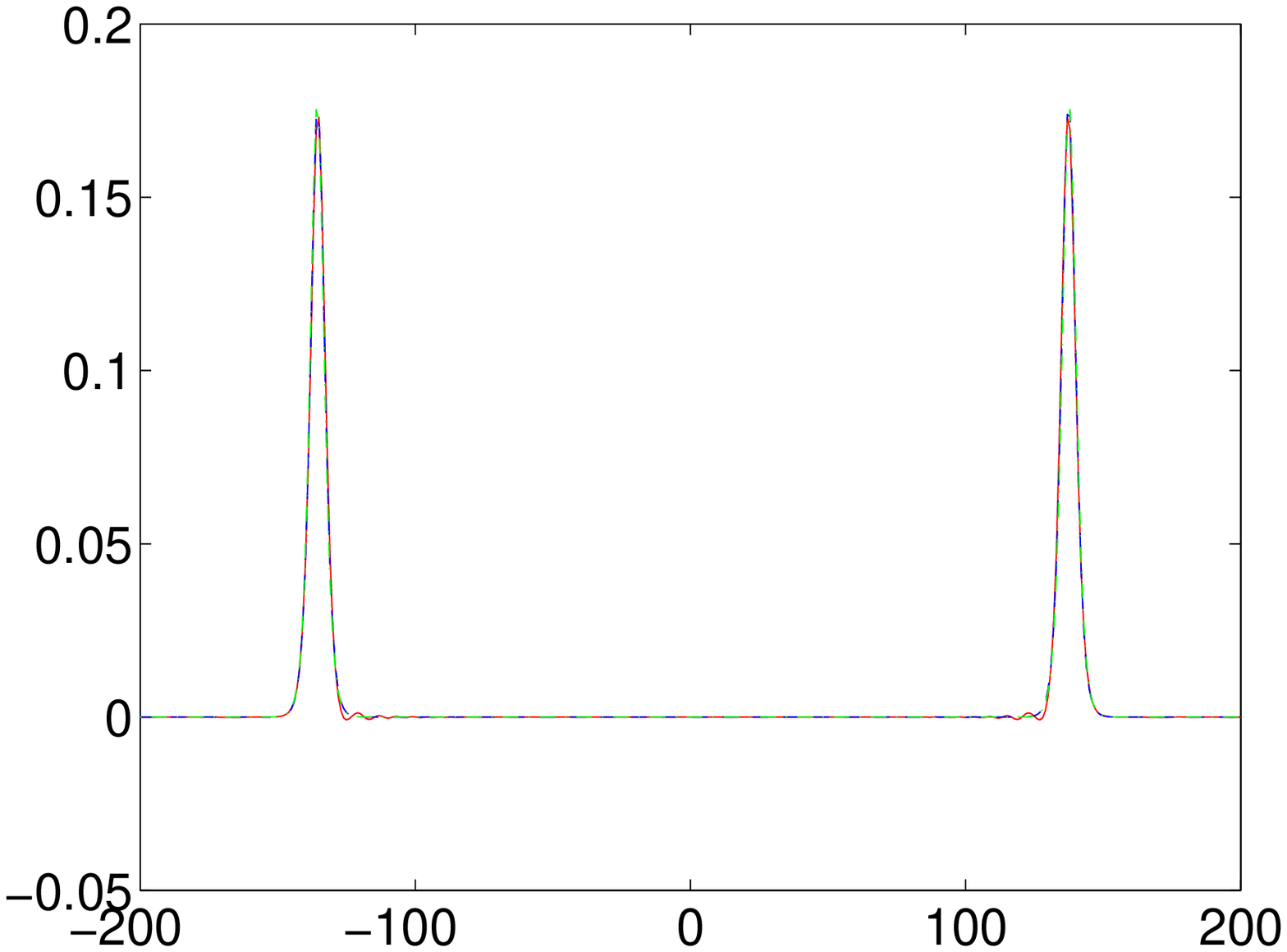}
	\end{center}
	\caption{(Color Online) In this example, the parameter values
$p = \frac{3}{2}$, $\delta_0=1$, $k=0.3$ are used. 
The initial condition consists of two solitons of the same amplitude at 
$-50$ and $50$. Here, the granular crystal
(non-integrable) dynamics is also compared to 
the mere addition of two one-soliton solutions of 
the Toda lattice. From top to bottom, left to right the snapshots shown 
are at $t=0, 30, 37, 40, 50, 150$. The (red) solid line is for 
the actual (numerical) granular lattice dynamics,  the (blue) dashed line is 
the 
plain superposition of two Toda one-soliton solutions of Eqn.~(\ref{TodaOne}), 
and the green dash-dotted line represents the 
numerical evolution of the Toda chain. \label{Toda1}}
\end{figure}

\begin{figure}[htbp]
	\begin{center}
	\includegraphics[width=.45\textwidth]{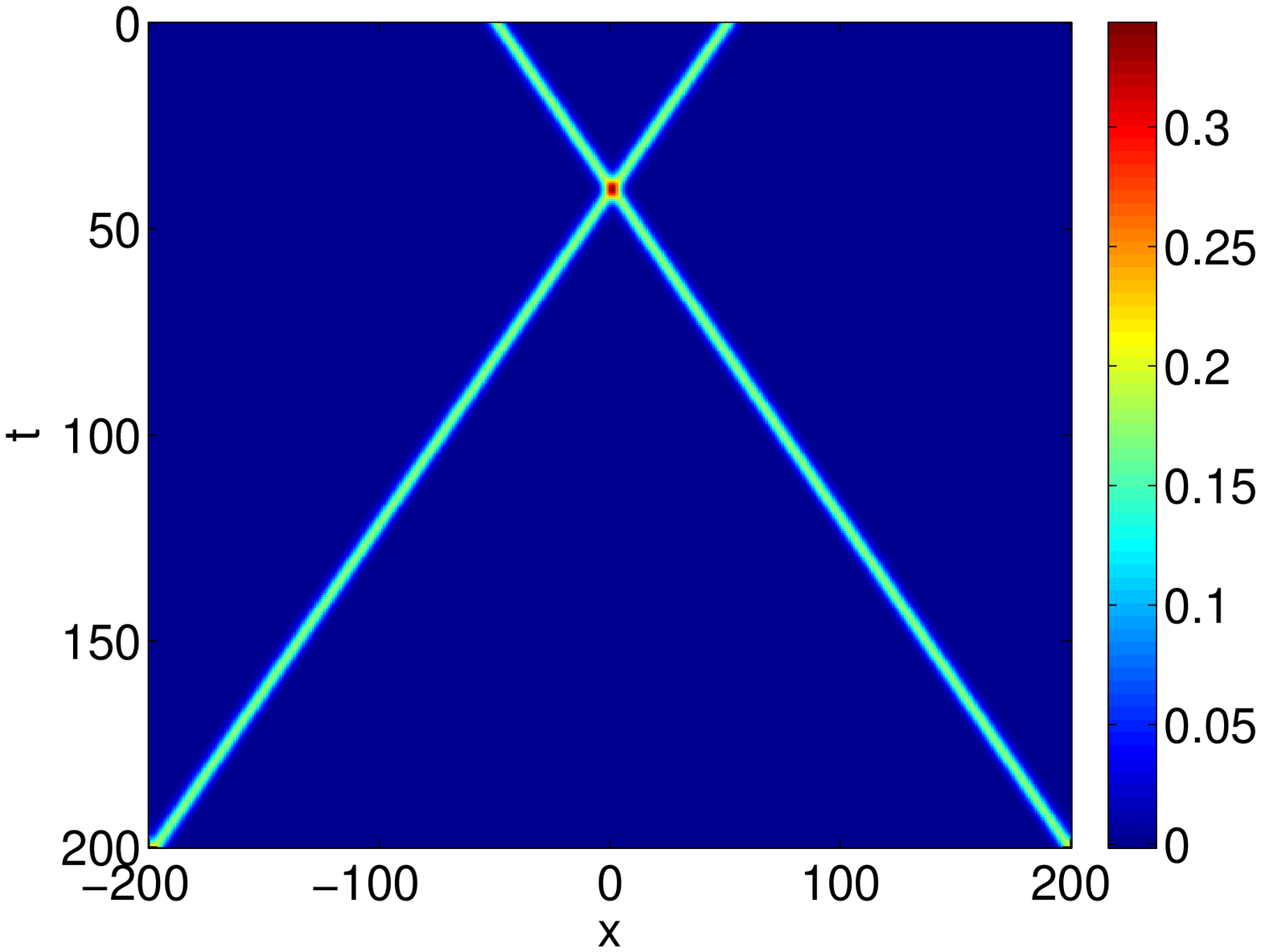}
	\includegraphics[width=.45\textwidth]{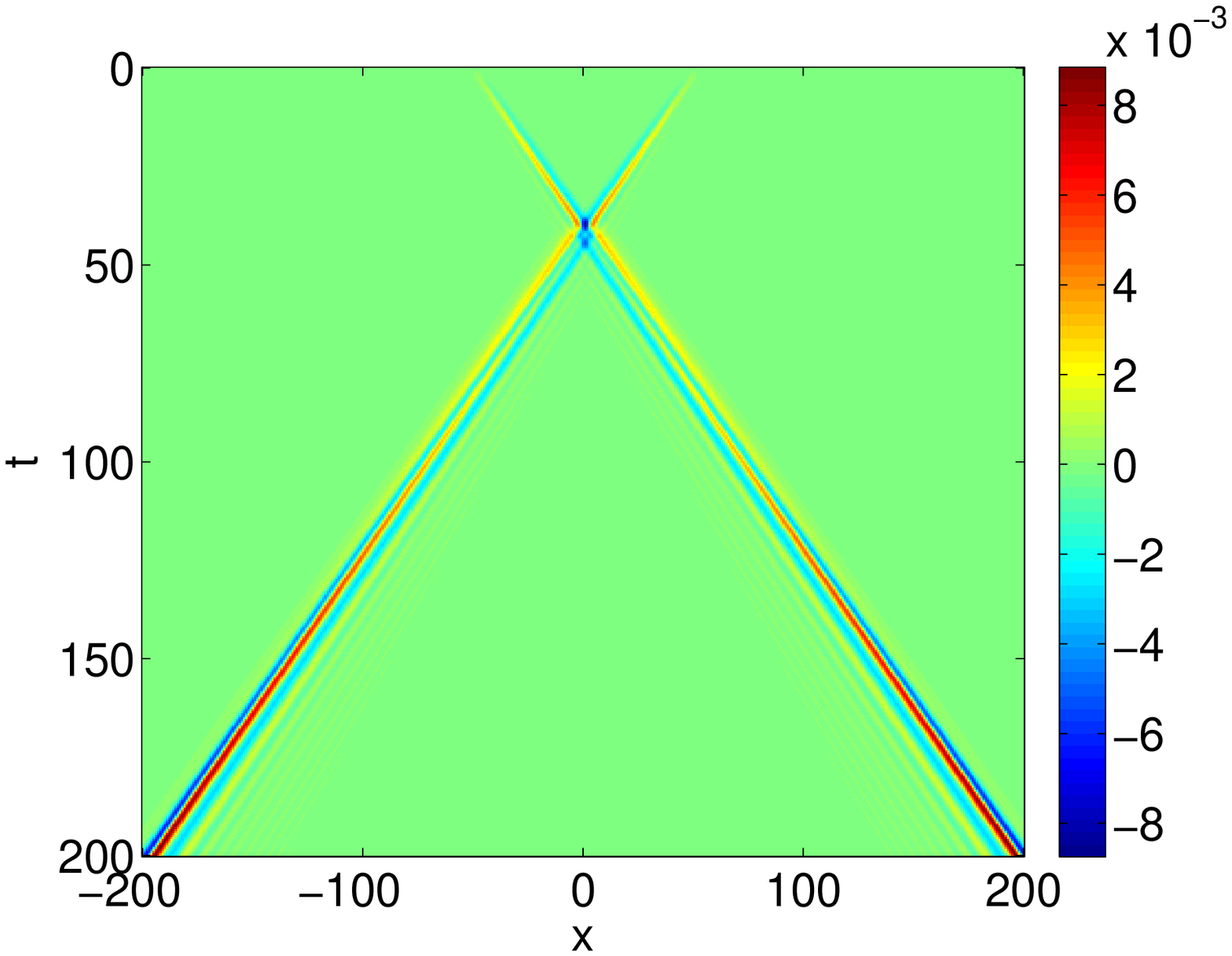}
	\end{center}
	\caption{(Color Online) In this contour plot space-time evolution of the
strains, the parameters and initial data are the same as those of 
Fig. \ref{Toda1}. The left panel represents the dynamical evolution of
two colliding solitary waves of 
the granular lattice. The right panel shows the {\it difference} 
between  the granular lattice and the superposition 
of two one-soliton solutions of the Toda lattice. The very small
magnitude of the difference (to be quantified further below) in the
colorbar in comparison to the left panel illustrates the relevance
of our approximation.}
\label{Toda2}
\end{figure}

The two-soliton solution of the Toda lattice is of the form \cite{toda2}
\begin{eqnarray}
{x}_n =S_{n-1}-S_n
\label{TodaTwoSoliton}
\end{eqnarray}
with 
\begin{eqnarray}
{S}_n =\ln \left\{ 1+A_1\exp[2(k_1 n - \beta_1 t) ]+A_2\exp[2(k_2 n -\beta_2t) ]+\exp[2(k_1+k_2) n -2(\beta_1+\beta_2)t ] \right\}
\label{toda_aux}
\end{eqnarray}
and
\begin{eqnarray}
\beta_i^2 = \sinh^2 k_i\\
A_1A_2 = \frac{(\beta_1+\beta_2)^2-\sinh^2(k_1+k_2)}{\sinh^2(k_1-k_2)-(\beta_1-\beta_2)^2}.
\end{eqnarray}
The results of this evolution are very similar to the ones
illustrated above 
and hence are not shown here.

In order to appreciate the role of the wave amplitude (and thus of the speed
in this mono-parametric family of soliton solutions) in the outcome
of the interaction, we have also explored 
 higher amplitude collisions, as shown in 
Figs.~\ref{TodaTwoSnap2}-\ref{TodaTwo2}. 
In these cases, the small amplitude wakes of radiation traveling
(at the speed of sound)
behind the supersonic wave  are more
clearly discernible. Nevertheless, once again the Toda lattice
approximation appears to capture accurately the result of such a
collision occurring at strain amplitudes of about half the precompression.
Fig.~\ref{TodaTwo2} again captures not only the granular chain evolution but also the relative error between that and the corresponding Toda lattice
evolution. Here, it is more evident that the eventual mismatch of speeds
of the waves between 
the approximation and the actual evolution yields a progressively larger
difference between the two fields.

As a systematic diagnostic of the  ``distance'' of the numerical 
granular crystal and approximate Toda-lattice-based solutions
(and as a check of our theoretical
prediction presented above), we have 
measured the $l_\infty$ norm (maximum absolute value in space and time) and 
the maximum of the $l_2$ norm in space of $(\tilde y_{n-1} - \tilde y_n)- (x_{n-1}-x_{n})$ till the two counter-propagating solitons are well separated.  
We measured this quantity as a function of the parameter $k$ 
(with $k=k_1=k_2$, $A_1 = A_2$) and report it as a function of 
the amplitude of $(x_{n-1}-x_{n})$ in Fig.~(\ref{Error_GLmToda}).
As shown in Fig.~(\ref{Error_GLmToda}), both graphs indicate a power law growth of the relevant error, with an exponent of $3.0010$ and $2.7557$ for the 
$l_\infty$ and $l_2$ norm of the error respectively.
These results can be connected with the  theoretical expectations
for this power law. In particular, as we saw above the theoretical
prediction for $\epsilon$ scales as $k^{5.5}$, while the amplitude, $A$,
of the solution is proportional to $k^2$, hence the scaling of
the quantity measured in our numerics is theoretically predicted
as $A^{2.75}$.  The close agreement with our numerics suggests that the theoretical estimate is tight  i.e.,  there is no normal form transformation which 
could push the residual between FPU and Toda to higher order.

\begin{figure}[htbp]
	\begin{center}
	\includegraphics[width=.4\textwidth]{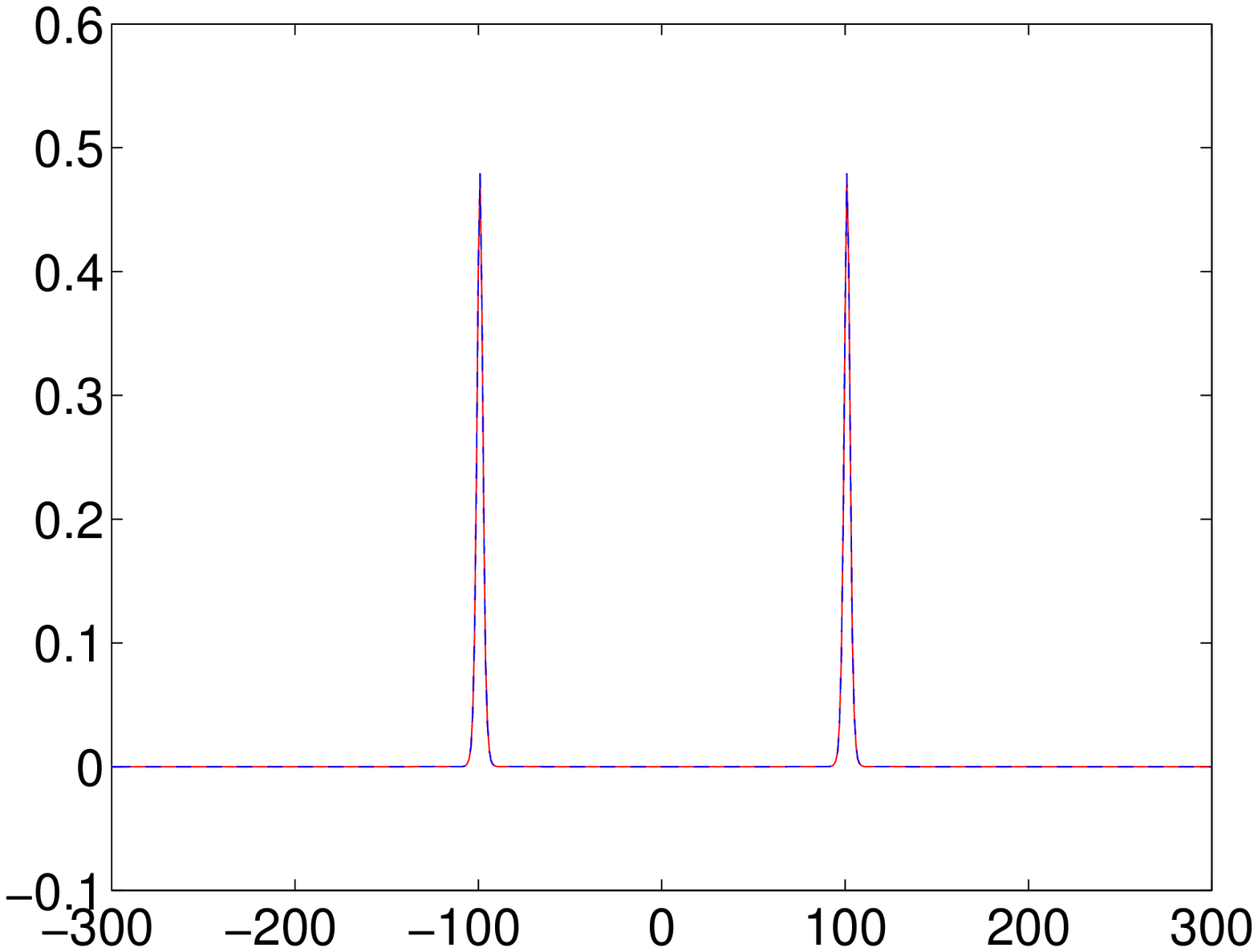}
	\includegraphics[width=.4\textwidth]{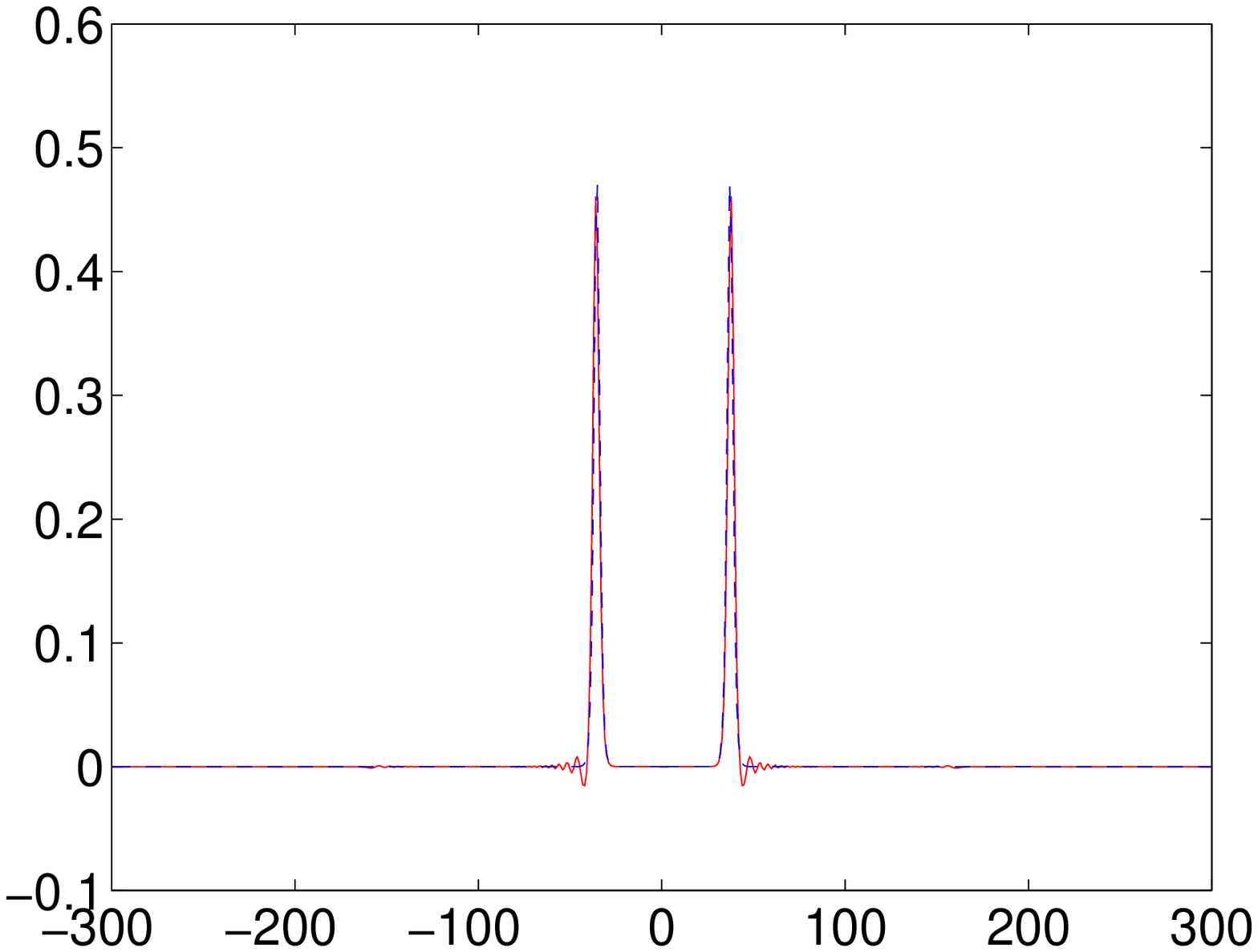}
	\includegraphics[width=.4\textwidth]{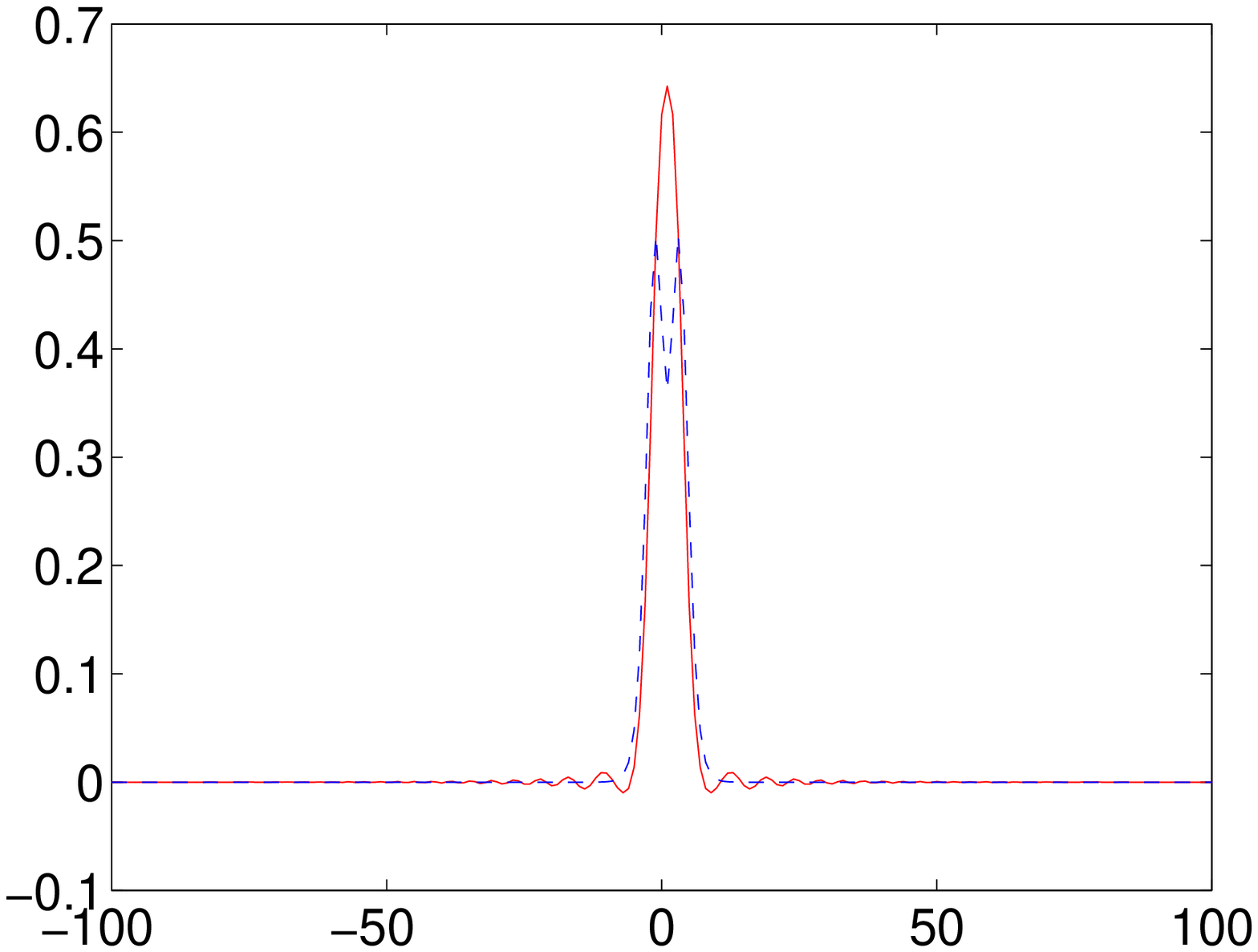}
	\includegraphics[width=.4\textwidth]{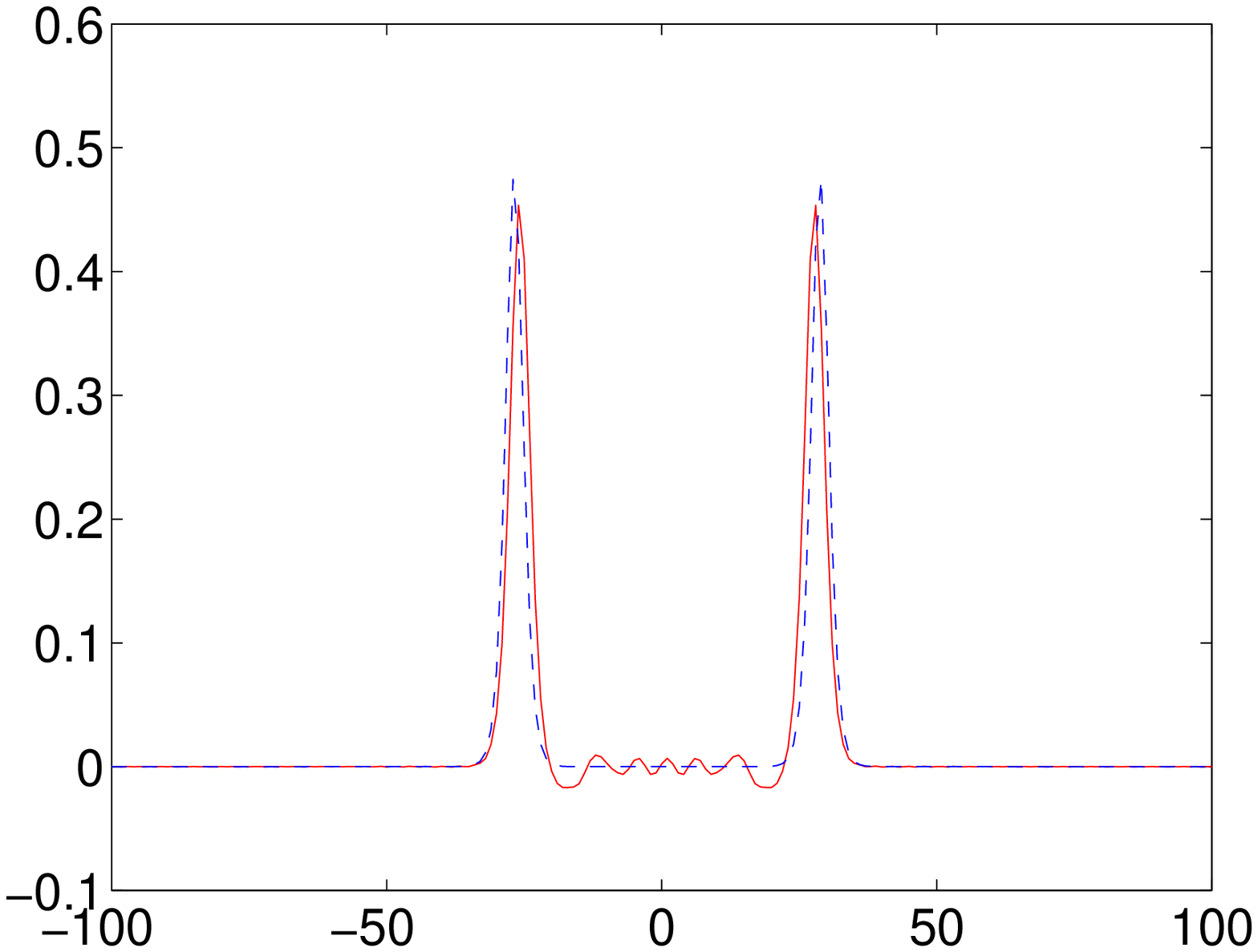}
	\includegraphics[width=.4\textwidth]{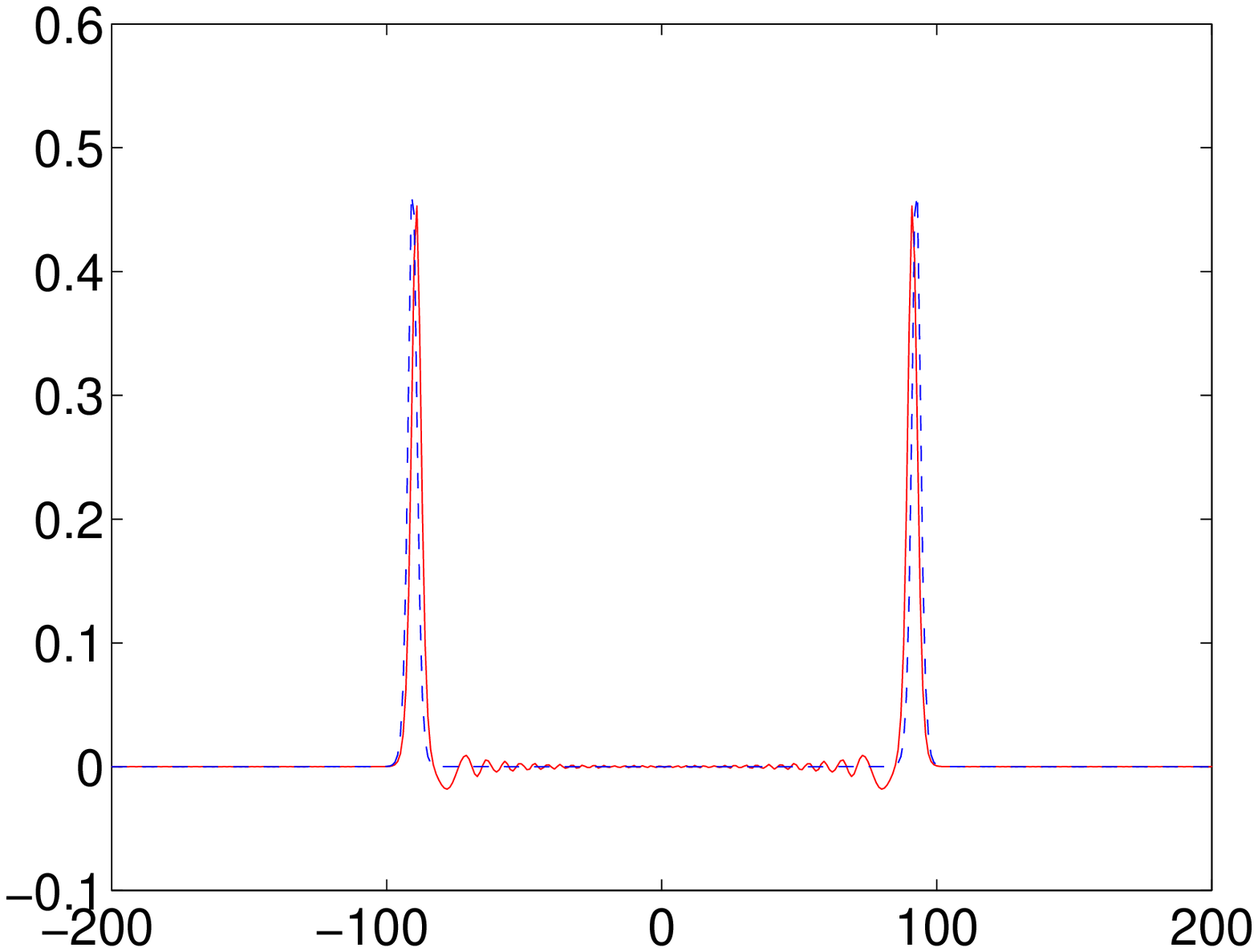}
	\includegraphics[width=.4\textwidth]{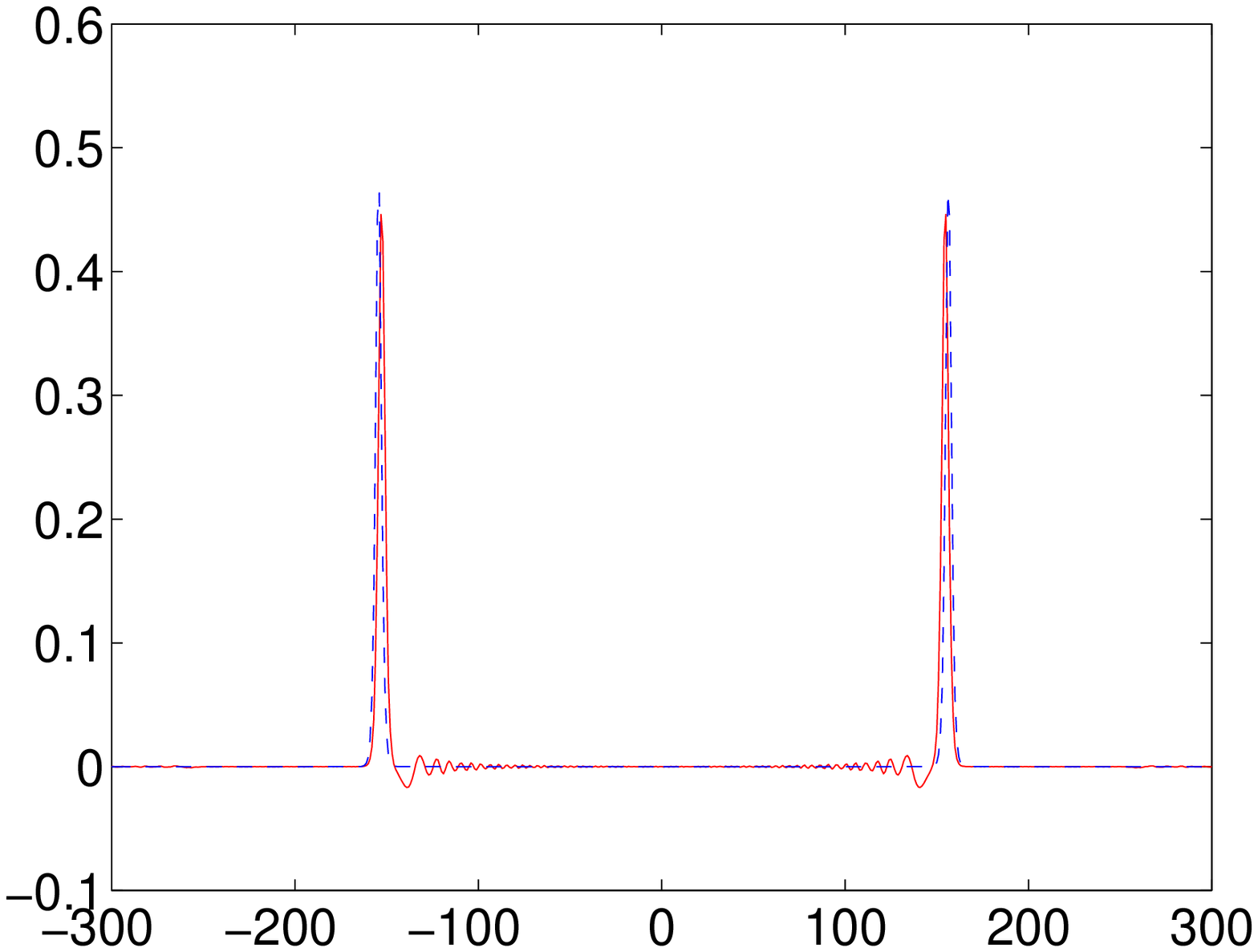}
	\end{center}
	\caption{(Color Online) In this dynamical evolution, the parameters
were chosen as $p = \frac{3}{2}$, $\delta_0=1$, $k_1=k_2 = 0.5$, $A_1=A_2=1.04$. The initial conditions consisted of a two-soliton solution with
waves of the same amplitude centered at $-100$ and $100$ 
at the Toda lattice level. 
From left to right, top to bottom, the snapshots at times 
$t=0, 50, 80, 100, 150, 200$ are shown. Once again, the solid (red) 
line denotes the numerical granular chain evolution 
dynamics while the dashed (blue) line 
stems 
from the exact Toda lattice two-soliton solution of  
Eq.~(\ref{TodaTwoSoliton}). The three curves are nearly coincident
for all the times considered.
\label{TodaTwoSnap2} }
\end{figure}

\begin{figure}[htbp]
	\begin{center}
	\includegraphics[width=.9\textwidth]{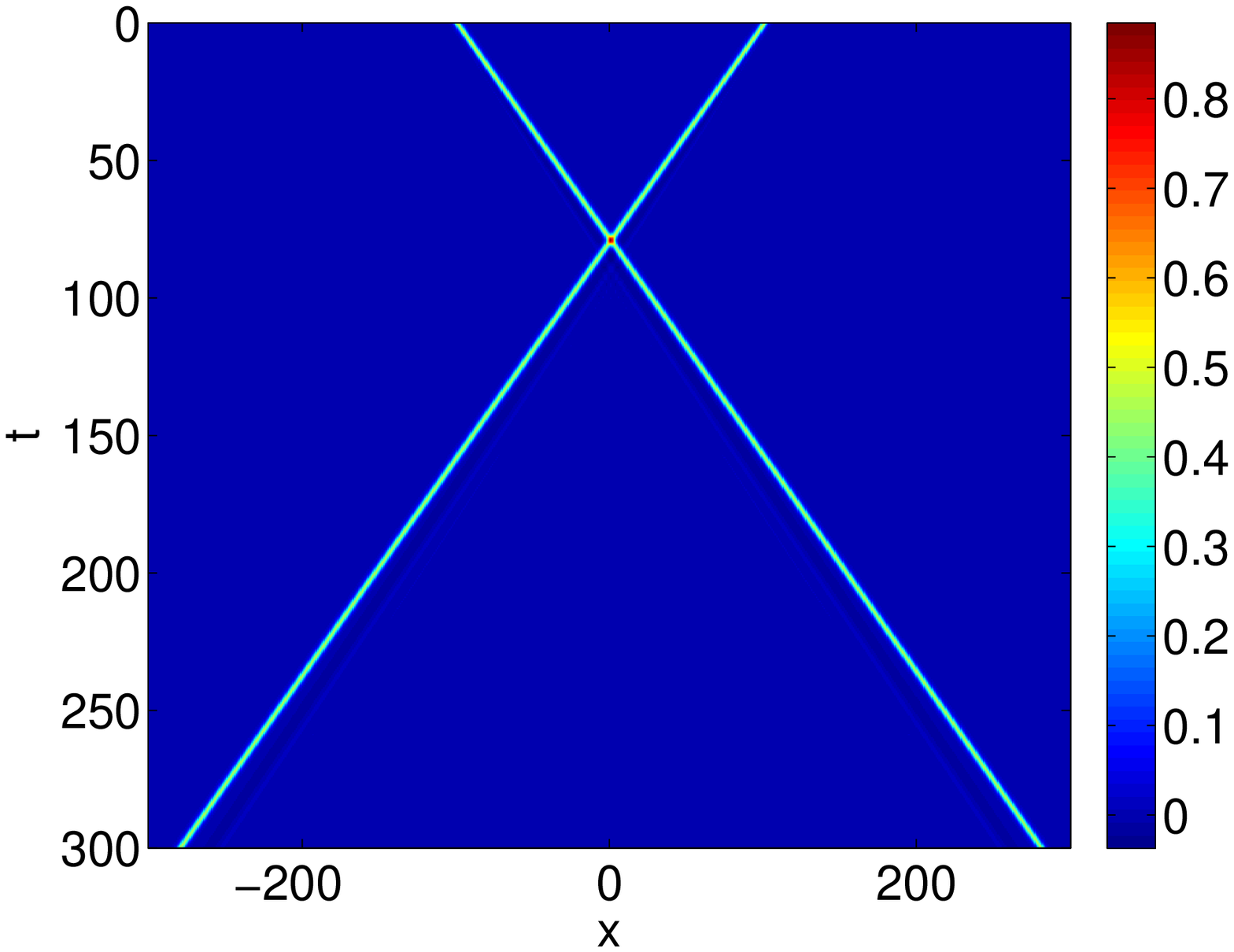}
	\includegraphics[width=.9\textwidth]{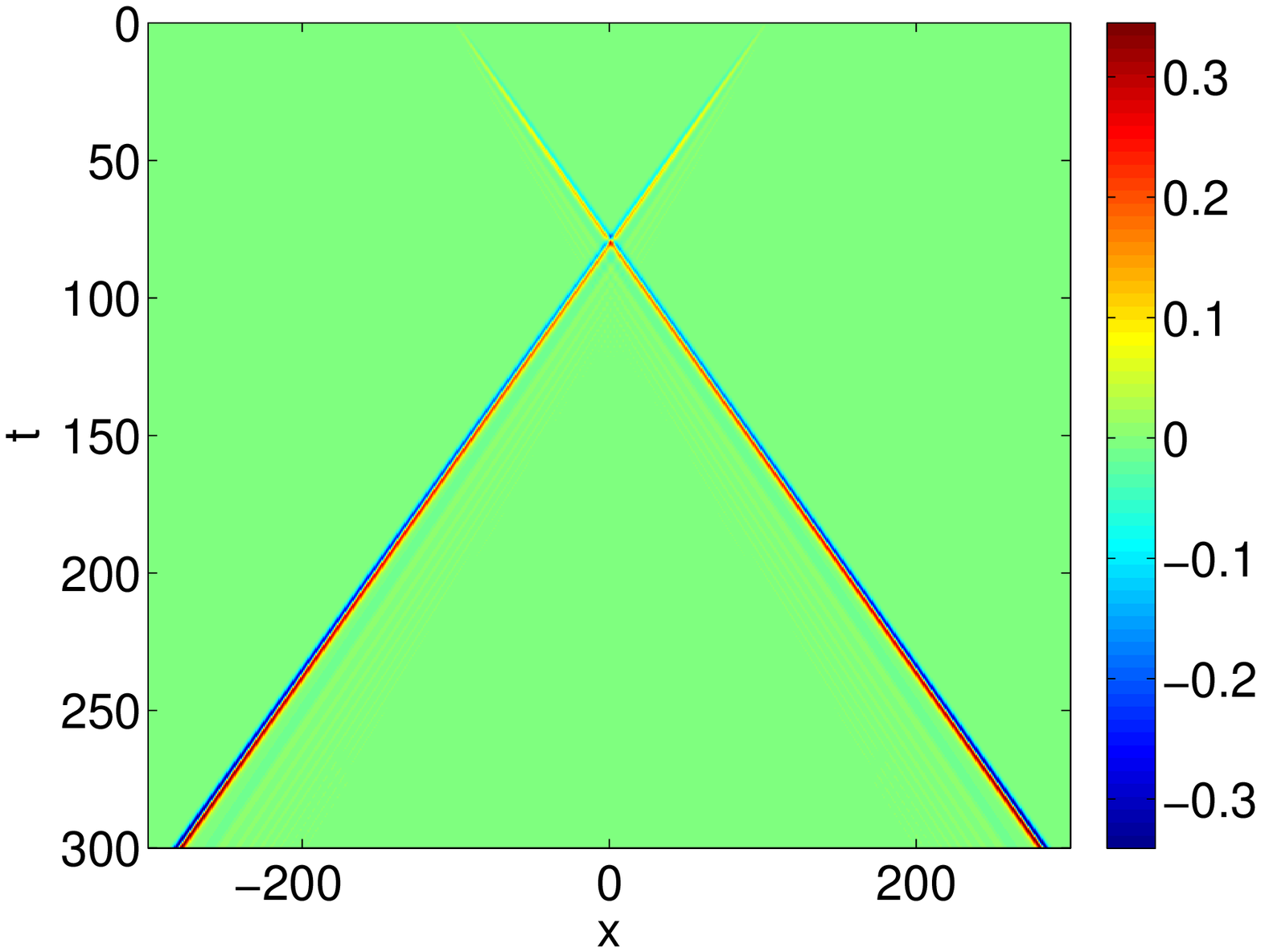}
	\end{center}
	\caption{(Color Online) 
The parameters and initial data have been chosen the
same as in Fig.~\ref{TodaTwoSnap2}, and the space-time contour plot
of the granular lattice evolution, as well as 
the difference of that from the two-soliton solution of Toda lattice
are shown.	
\label{TodaTwo2}}
\end{figure}

\begin{figure}[htbp]
	\begin{center}
	\includegraphics[width=.4\textwidth]{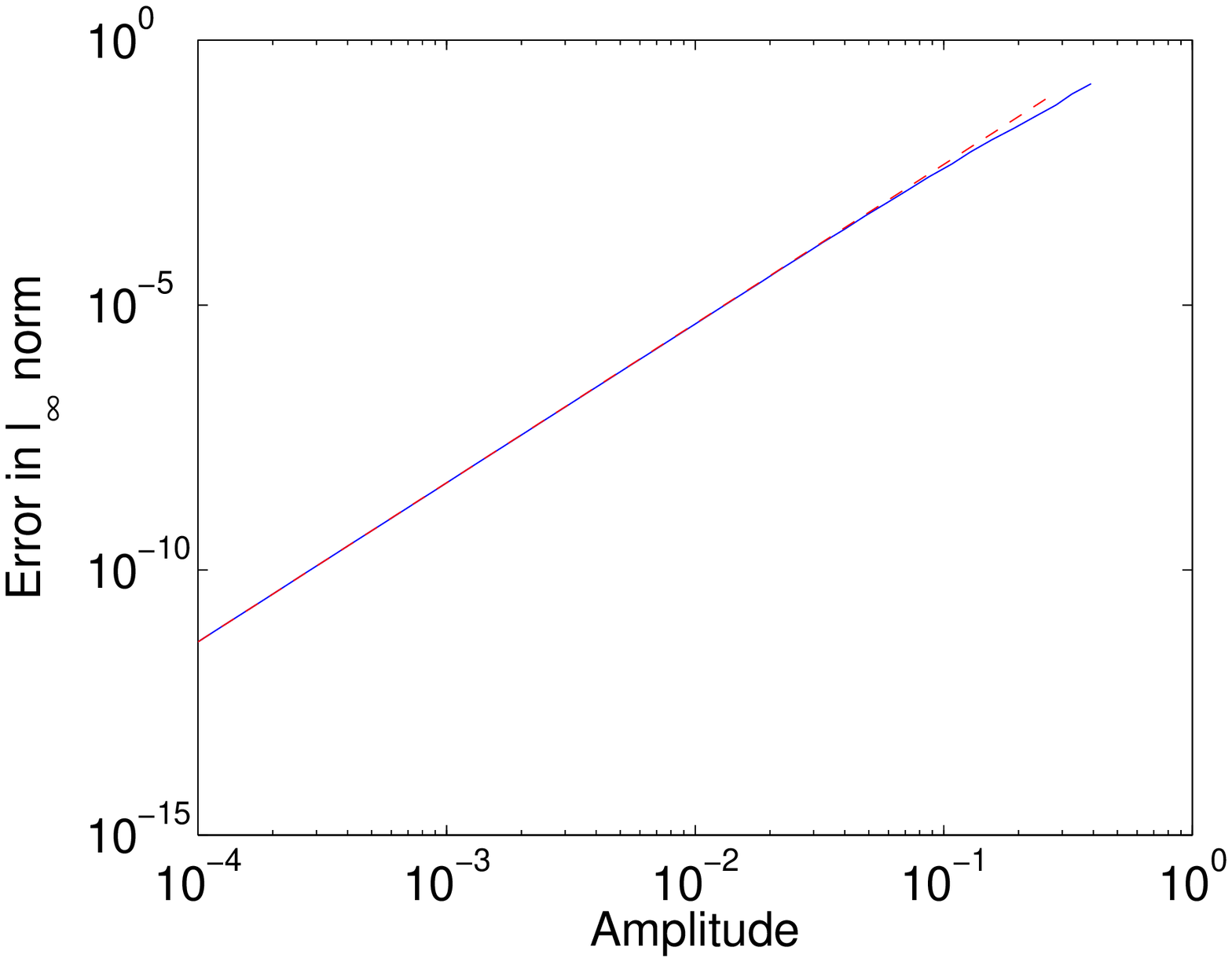}
	\includegraphics[width=.4\textwidth]{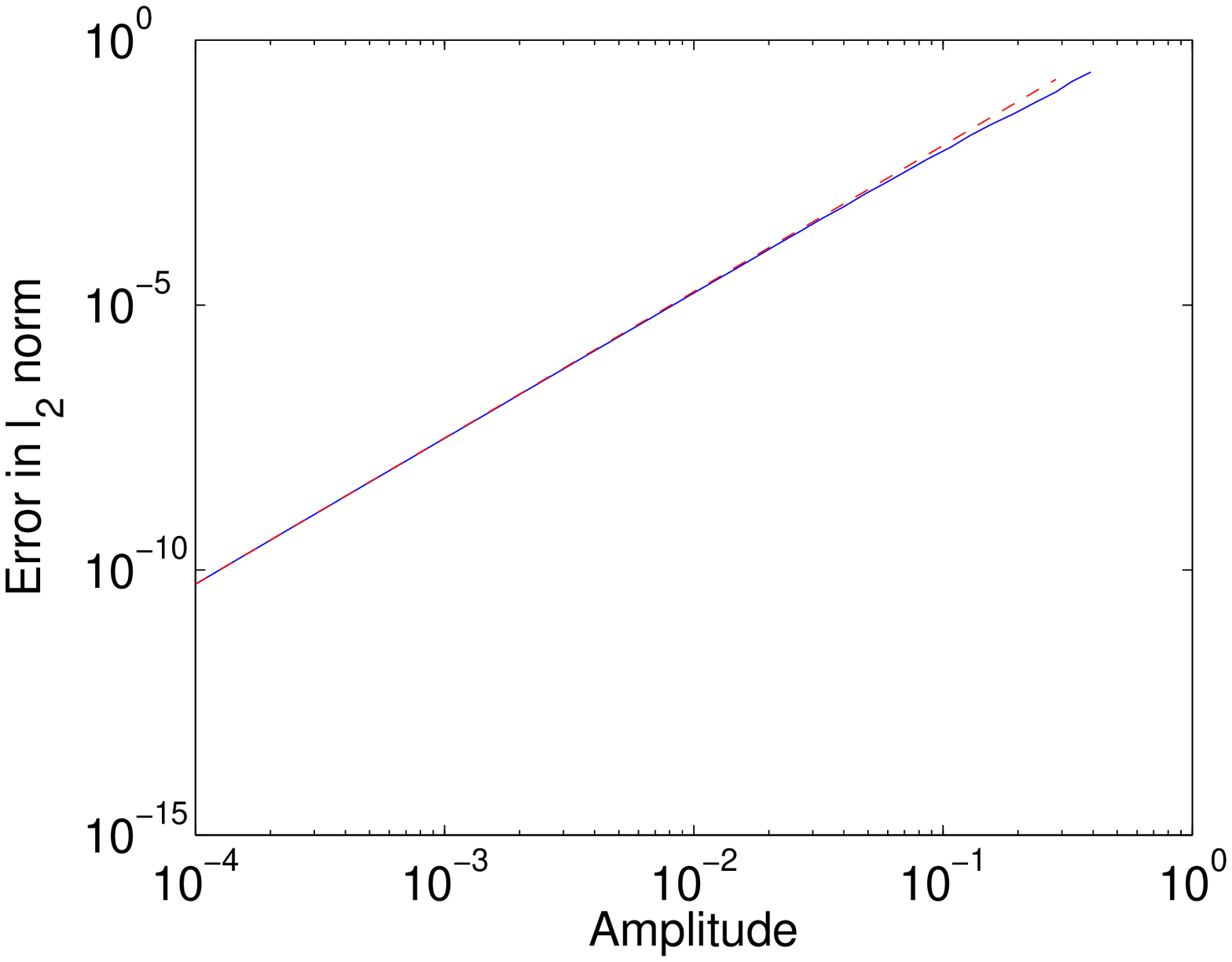}
	\end{center}
	\caption{(Color Online) The left panel shows the $l_\infty$ norm of the error 
(i.e., difference of granular evolution from the Toda lattice
2-soliton solution) until the two solitons are well separated after 
the collision, versus the amplitude of the initial data. 
The right panel is the $l_2$ norm of the same quantity. 
Both clearly represent a power law with a best fit exponent of 
$3.0010$ and $2.7557$ respectively (shown in red dash line).
\label{Error_GLmToda}}
\end{figure}

Lastly, we explore the case of the Toda lattice approximation for
the case of two co-propagating solitary waves. I.e., recalling that
one of the advantages of the Toda lattice approximation is not only
its discrete nature, but also its ability to capture both co-propagating
and counter-propagating solutions,
we use the 2-soliton solution of Toda lattice of the form \cite{toda1}
\begin{eqnarray}
 {x}_{n-1}-x_n =S_{n-1}-2S_n+S_{n+1}
\label{TodaTwoSoliton2}
\end{eqnarray}
with 
\begin{eqnarray}
{S}_n =\ln \left\{ \cosh[k_1 (n-n_1)- \beta_1 t ] +B \cosh[k_2(n-n_2)-\beta_2 t ]\right\}.
\end{eqnarray}
For the waves propagating  in the same direction
\begin{eqnarray}\label{SD}
\beta_1 = 2\sinh\frac{k_1}{2}\cosh\frac{k_2}{2}\\
\beta_2 = 2\sinh\frac{k_2}{2}\cosh\frac{k_1}{2}\\
B = \sinh\frac{k_1}{2}/\sinh\frac{k_2}{2},
\end{eqnarray}
and the result of a typical example of the dynamical evolution
is shown in Figs.~\ref{TodaTwoRRSnap}-\ref{TodaTwoRR}. 
It can be clearly observed here that 
the 
co-propagating case yields a far less accurate description
than the counter-propagating one. This is presumably because
of the shorter (non-integrable) interaction time in the latter
in comparison to the former. 
Furthermore notice that again
the disparity between the two evolutions 
is far more pronounced for large amplitude waves,
as the small amplitude one is accurately captured throughout
the collision process. Nevertheless, once again our integrable
approximation is quite useful in providing at least a
qualitative, essentially analytical handle on the interaction dynamics
observed herein.

\begin{figure}[htbp]
	\begin{center}
	\includegraphics[width=.4\textwidth]{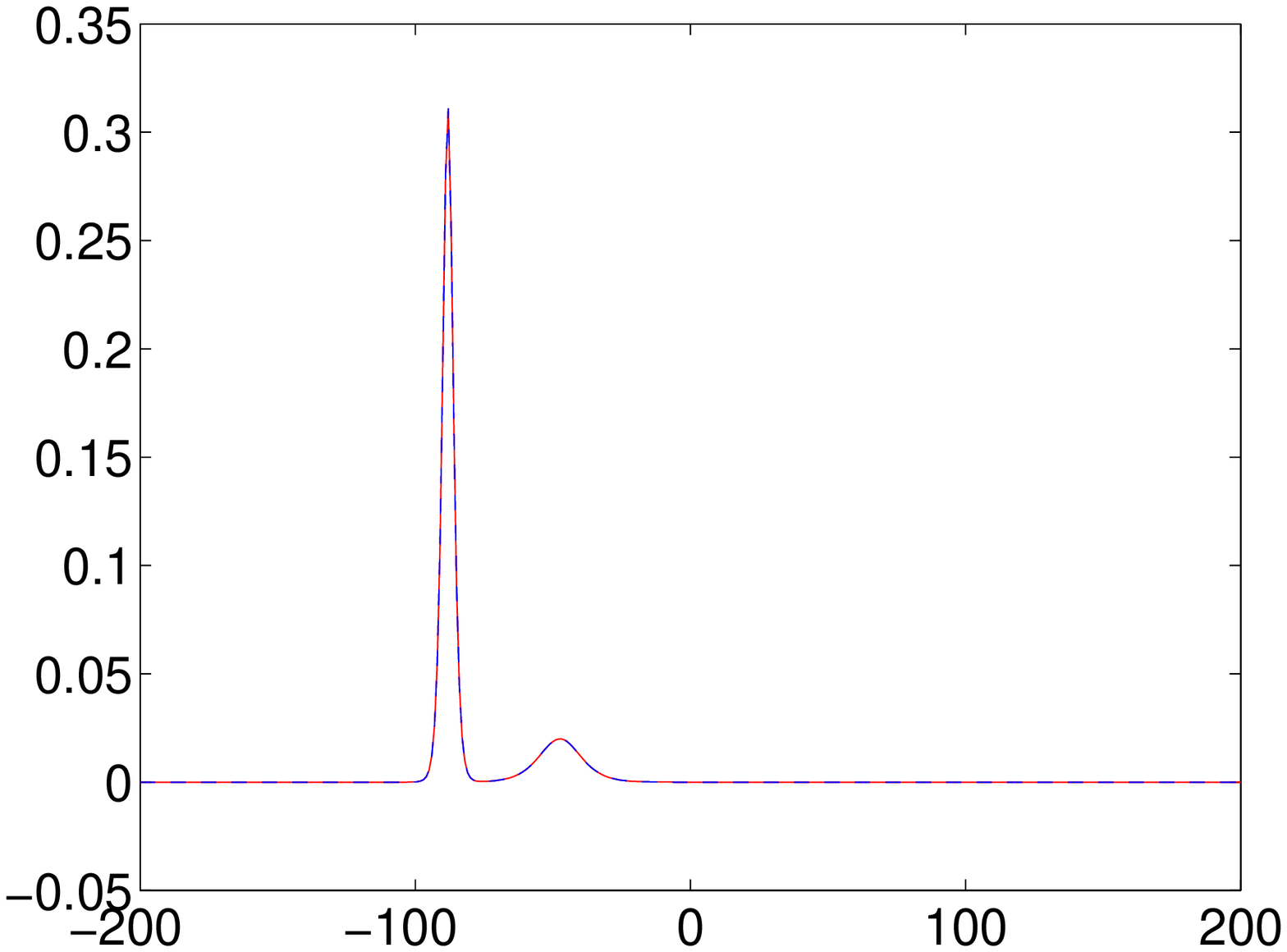}
	\includegraphics[width=.4\textwidth]{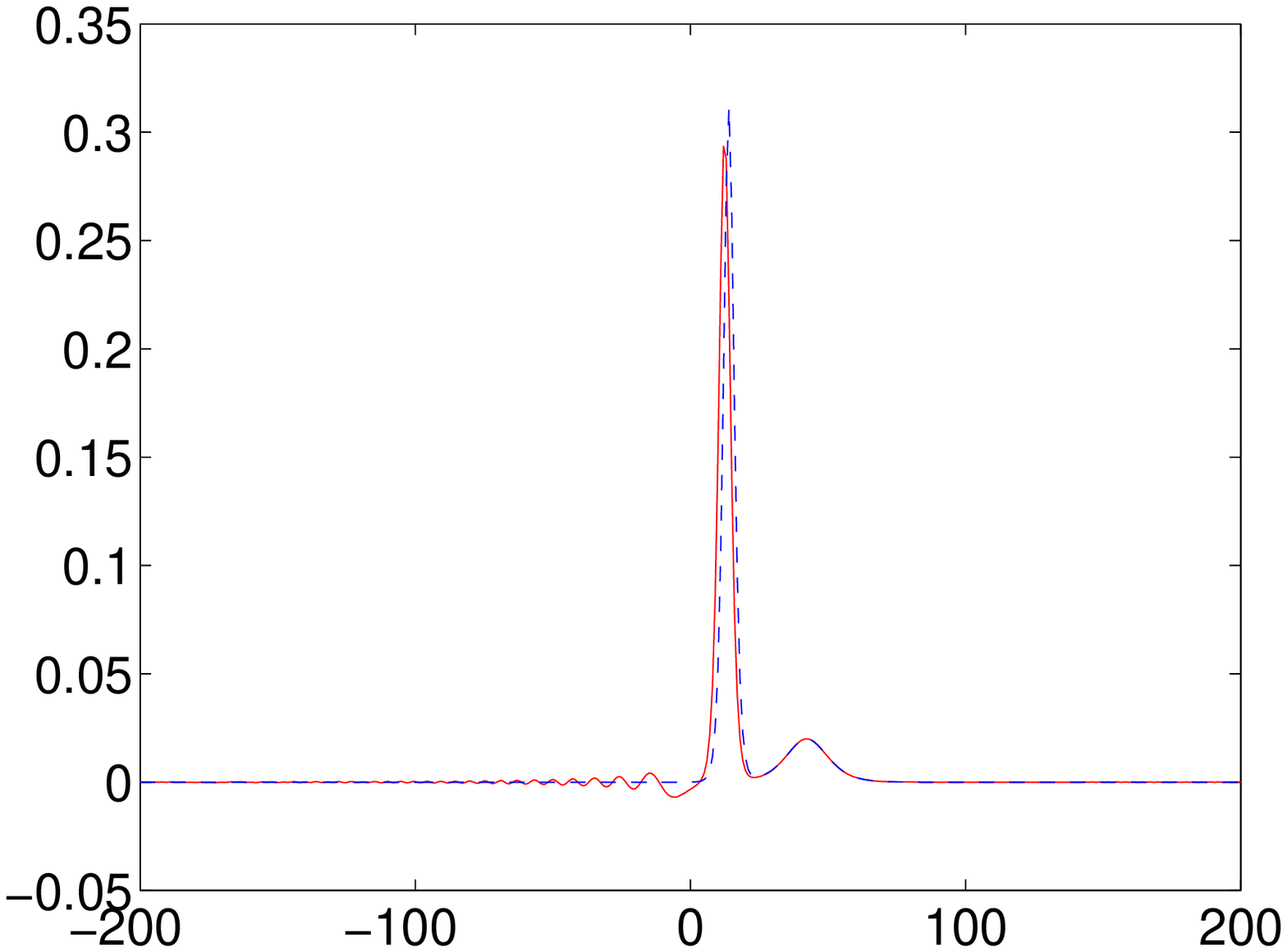}
	\includegraphics[width=.4\textwidth]{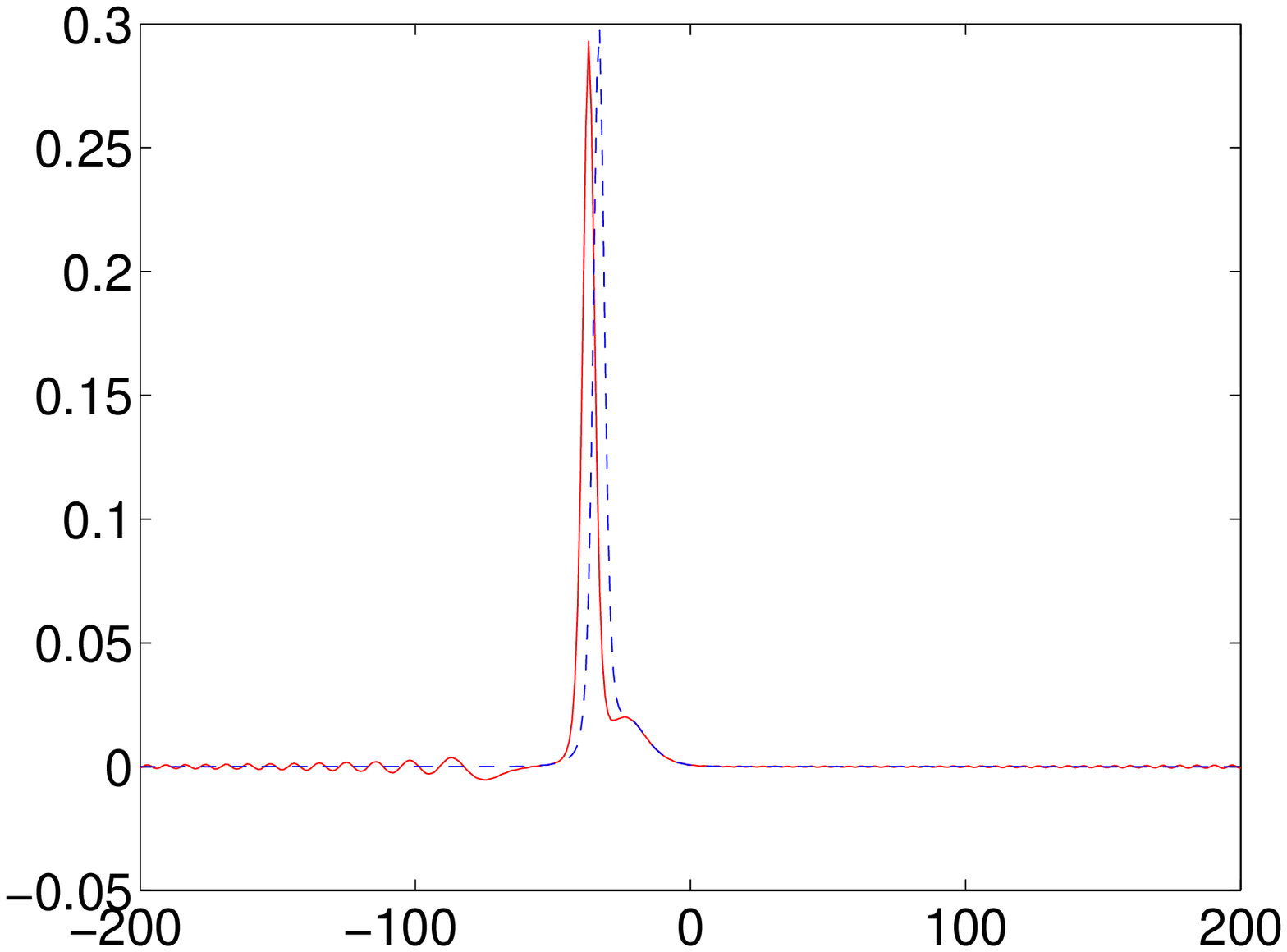}
	\includegraphics[width=.4\textwidth]{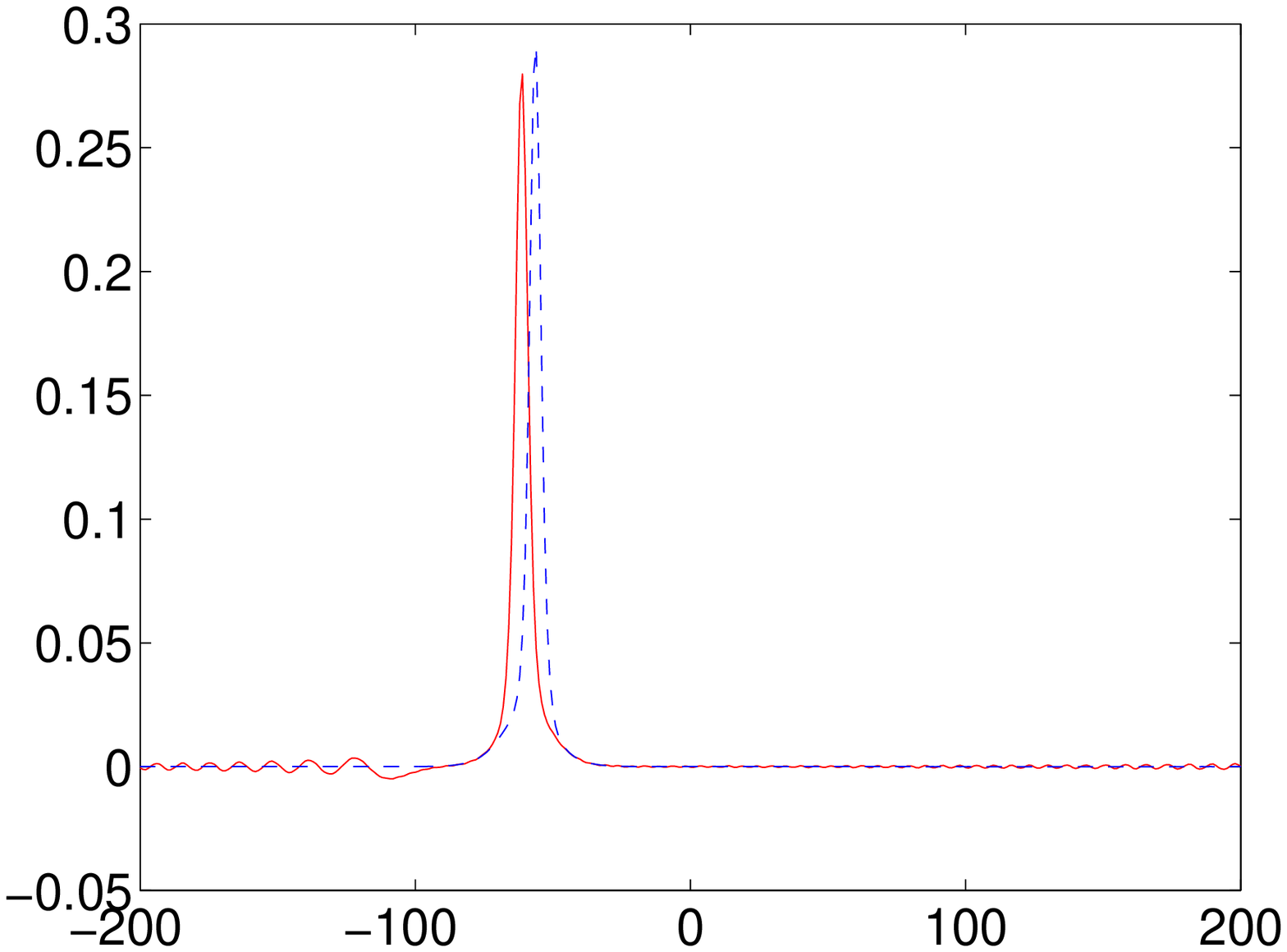}
	\includegraphics[width=.4\textwidth]{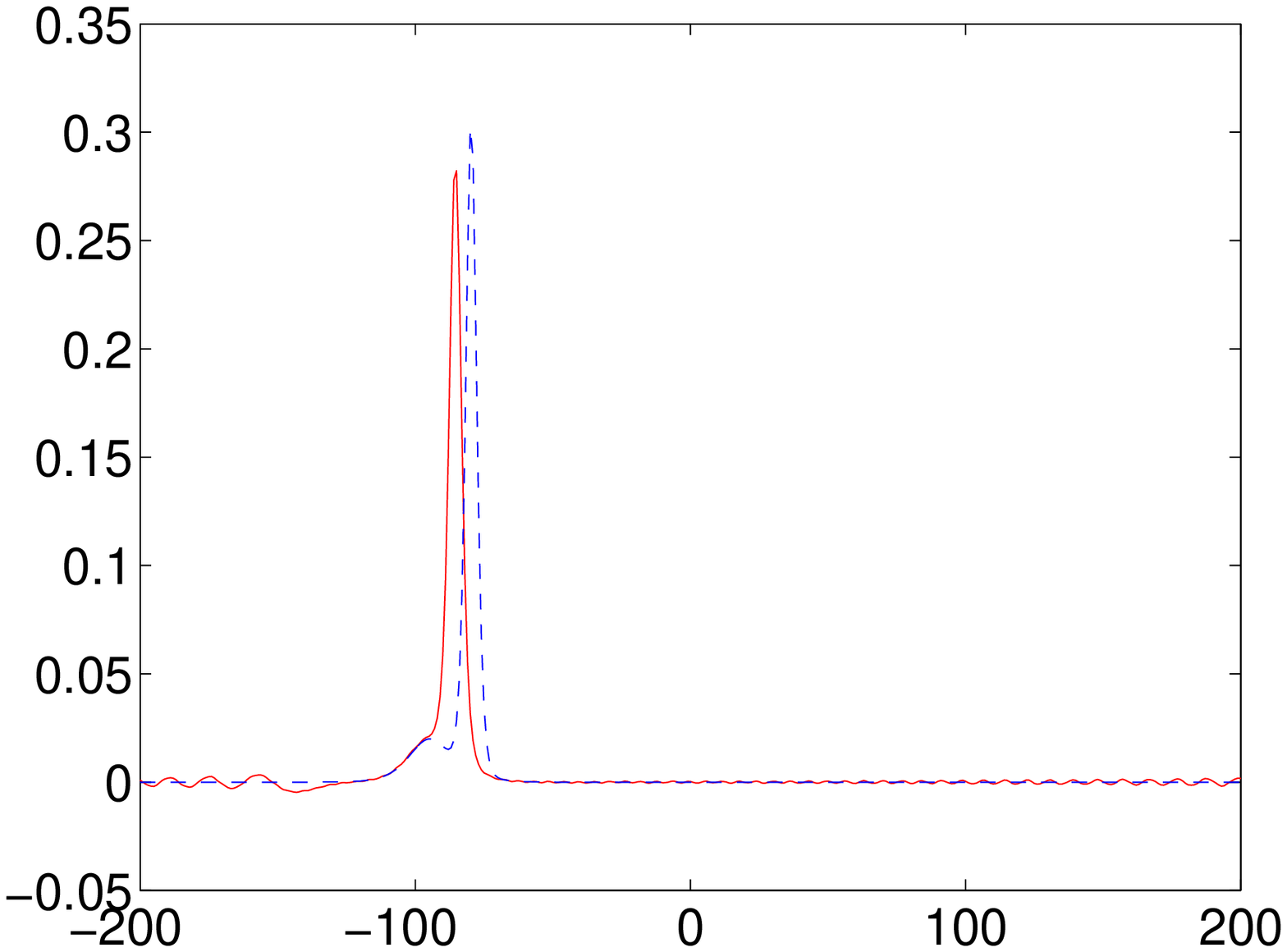}
	\includegraphics[width=.4\textwidth]{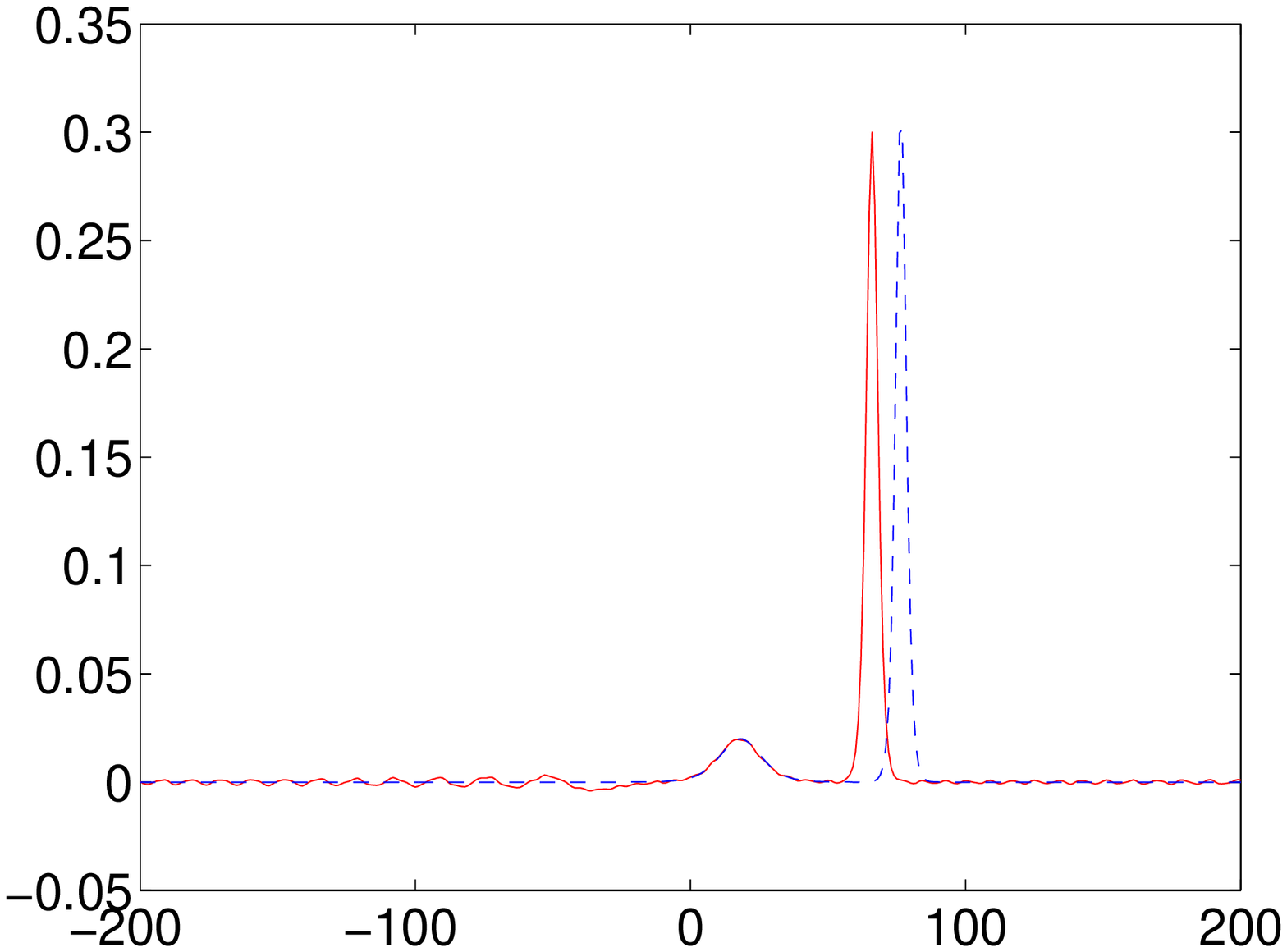}
	\end{center}
	\caption{(Color Online) The case of $p = \frac{3}{2}$, $\delta_0=1$, $k_1=0.3$, $k_2=0.5$, with an initial condition of two solitons at $-100$ and $-80$ 
propagating in the same direction (to the right) is shown. From left to 
right, top to bottom, the snapshots of 
$t=0, 400, 1000, 1300, 1600, 3000$ are shown. The (red) solid line 
represents the granular lattice evolution and the (blue) dashed line the 
Toda 2-soliton solution of Eq.~(\ref{TodaTwoSoliton2}). \label{TodaTwoRRSnap} }
\end{figure}

\begin{figure}[htbp]
	\begin{center}
	\includegraphics[width=.9\textwidth]{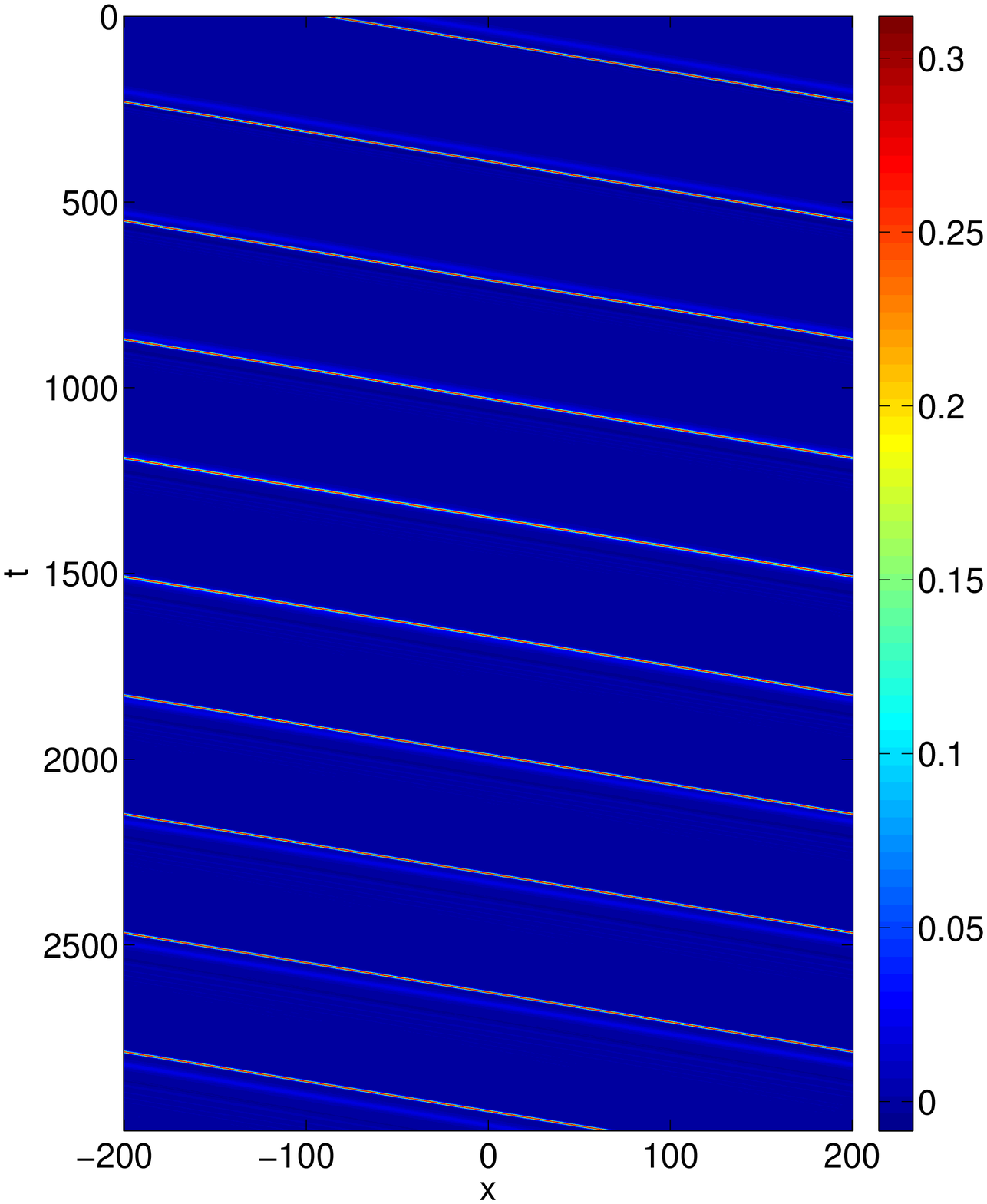}
	\end{center}
	\caption{(Color Online) Same as the previous figure but now through the
space-time evolution shown in a contour plot.
\label{TodaTwoRR}}
\end{figure}

\section{Conclusions and Future Challenges}

We believe that the present work provides an insightful and meaningful
(not only qualitative but even quantitatively, where appropriate)
way for considering
the interactions of solitary waves in the 
realm of granular crystal dynamics. 
What makes this work particularly timely 
and relevant is that the granular chain problem 
is currently both theoretically
interesting and experimentally, as well as computationally tractable. 
Two types of approximations were proposed
herein for developing a qualitative and even semi-quantitative understanding
of such collisions. The first was based on the well known  KdV equation. While
this is a useful approximation, some of its limitations were discussed,
the 
most notable being the continuum, long-wavelength nature of the approximation,
as well as the uni-directional character of the interaction (i.e., 
co-propagating waves). To avoid these constraints, a second approximation,
based on the Toda lattice was also presented. The latter provided
a high-quality description especially of counter-propagating wave
collisions, while its 2-soliton solutions can also be used for
capturing co-propagating cases, at least qualitatively. It is relevant to note
here that at the level of $\delta_0=0$ (i.e., when the precompression
is absent) such collisions have {\it already} been studied
in~\cite{sen1,sen00,jobmelo}. It thus seems 
that an extension of that work to experimentally consider 
collisions in the more
theoretically and analytically
tractable case of finite precompression would be possible, as
much as it would be desirable.

We believe that this line of thinking, and especially the approximation
of using a discrete model such as the Toda lattice  could provide
a useful tool for  understanding different forms of solitary wave 
interactions in Hertzian 
systems~\cite{Theocharis2009,Job2009,Boechler2010,vakakis2,vakakis3}. 
A more
ambitious generalization would involve the consideration
of two-dimensional lattices and the potential reduction
thereof to Kadomtsev-Petviashvilli continuum models (i.e.,
2d generalizations of the KdV) or perhaps to other lattice models in order 
to understand the dynamics of higher dimensional such chains.
Another challenging problem would be to obtain some analytical
understanding of the collisions without precompression; the difficulty
in that case stems from the absence of a well-established,
yet analytically tractable (discrete or continuum) description
for capturing multi-soliton interactions.
Such directions are currently under consideration and will
be presented in future publications.

\appendix
\section{Estimates in the small amplitude regime}
To obtain rigorous estimates it is first useful to write the general FPU chain as a first order system.
Observe that the chain of oscillators
\[ \ddot{q}_n = V'(q_{n-1}-q_n) - V'(q_n - q_{n+1}) \]
can be rewritten as the system
\begin{equation} \left\{ \ba{l} \dot{r}_n = p_{n-1}-p_n \\ \dot{p}_n = V'(r_n) - V'(r_{n+1}) \ea \right. \label{FPU} \end{equation}
upon making the change of variables $p_n = \dot{q}_n$, $r_n = q_{n-1}-q_n$.  Both the granular chain and the Toda lattice are special cases with $V'(q) = (\delta_0 + q)^p$ and $V'(q) =e^q$ respectively.  Writing the Hamiltonian $H(r,p) = \sum_n \frac{1}{2}p_n^2 + V(r_n)$ and the operator $J$ defined by $[J(r,p)]_n = (p_{n-1}-p_n, r_n-r_{n+1})$ the equation \eqref{FPU} is rewritten as the system of Hamiltonian ODEs 
\[ \frac{d}{dt}\left(\ba{c} r \\ p \ea \right) = JH'\left( \left(\ba{c} r \\ p \ea \right) \right). \]

Notice that in the above expression 
$[H'\left( \left(\ba{c} r \\ p \ea \right) \right)]_n = \left(\ba{c} V'(r_n) \\ p_n \ea \right)$.
Having provided this general setup for our Toda and granular
Hamiltonian chains, we now proceed to present the proposition that
estimates the proximity between the 2-soliton solutions of these two models
for small amplitude initial data (quantified by
$\delta$ below) and long times (quantified by $k^{-1}$ below).

\begin{proposition}
Let $H_*$ denote the Toda Hamiltonian and let $\left(r_*(t;k_1,k_2,\gamma_1,\gamma_2),p_*(t;k_1,k_2,\gamma_1,\gamma_2)\right)$ denote its four-parameter family of two-soliton solutions.  Let $V$ be a general smooth interaction potential satisfying $V(0) = V'(0) = 0$ as well as $V''(0) > 0$ and $V'''(0) \ne 0$ and let $(r,p)$ denote the solution of the corresponding FPU lattice with initial condition $(r_*(0),p_*(0))$.  

There is a $\delta > 0$ such that so long as $|k_j| < \delta$ for $j = 1,2$, then for any $\eta \in (0,1)$ the estimate 
\[ \|r(t)-r_*(t)\| + \|p(t)-p_*(t)\| < Ck^{5.5-2 \eta} \]
holds for $0 \le t \le k^{-(1+\eta)}$.
\end{proposition}

In the case of counterpropagating solitary waves, the time scale $k^{-1}$ is sufficiently long for the waves to pass through each other.  The content of the theorem is that the amount of energy that is transferred from coherent modes to radiative modes is very small compared to the energy in the coherent modes.

The proposition can be regarded as a corrolary of three lemmas.
Before we embark into their technical description, let us give a
brief outline of the physical significance of each one.
Lemma 3 below shows that in the context of the Toda 2-soliton solution, 
the interaction of two broad shallow waves does not produce high 
(i.e. order greater than $k$) frequency ripples i.e., ``radiation''
corresponding to such wavenumbers.
Lemma 1 makes use of lemma 3 to quantify the ``local truncation error'' 
for the scheme given by evolving a Toda 2-soliton  in lieu of FPU.
I.e., when we evolve with Toda 2-soliton initial conditions 
in our FPU (non-integrable) lattice instead of the
integrable Toda one, there is a local truncation error stemming 
from the difference between the two lattice dynamics. This lemma
quantifies this difference as a function of the solution amplitude
(represented by $\delta$). Finally, lemma 2 estimates how the difference
of our FPU-type lattice and the Toda lattice evolves over time
on the basis of the above local truncation error and how the latter 
``accumulates'' over a long interval of time $T$ (characterized by $k^{-1}$).

\begin{lemma}
Let $x_*$ denote a Toda 2-soliton solution. Let $J$, $H$ and $H_*$ be defined as above.  Define $F(x) = JH'(x)$ and define $F_*(x) = JH_{*}'(x)$.
Let $\kappa_1$ and $\kappa_2$ be fixed numbers and let the amplitude parameters for $x_*$ be given by $k_1 = \delta \kappa_1$ and $k_2 = \delta \kappa_2$.  

There is a $\delta_1$ so that for all positive $\delta < \delta_1$ the following hold: 

\[ \| F'(x_*) - J\| \le C\|Jx_*\| \le C \delta^{2.5} \]
\[ \|F''(x_*) \| \le 1 + \sup_n V'''(x_*) \le C \]
\[ \| F(x_*) - F_*(x_*) \| \le C\|Jx_*^3\| \le C\|Jx\| \|x\|_\infty^2 \le C \delta^{6.5} \]
\end{lemma}

Here the half powers of $k$ arise because of the slow decay in the tails of the solitary wave.  More explicitly a solitary wave satisfying $x_n \sim  C k^2 e^{-kn}$ satisfies $\|x\|^2 \le \frac{Ck^4}{1-e^{-2k}} \sim k^3$ and similarly $\|J^n x^m\| \le C k^{m+\frac{n}{2}}$.

\begin{lemma}
Let the following be given: A Hilbert space $\mathcal{H}$ with inner product $\langle \cdot, \cdot \rangle$ and associated norm $\| \cdot \|$, an open subset $U \subset H$, $C^3$ functions $F$ and $F_*$ from $U$ to $\mathcal{H}$, with identical and symplectic linear part, i.e.  $J := F'(0) = F_*'(0)$ satisfying $\langle Jv,v\rangle = 0$ for all $v \in \mathcal{H}$.    

There exist positive constants $\delta_0$ and $C_0$ such that if the estimates
\[ \sup_{t \in [0,T]} \|F'(x_*)-J\|_{H \to H} < \frac{\delta_0}{T},\] 
\[\sup\{ \|F''(v)\|_{H \to \mathcal{L}(H)} \; | \; \|v\|_H \le 2\|x_*\|\} < C_0 \]
and
\[ \sup_{t \in [0,T]} \|F(x_*)-F_*(x_*)\| < \frac{1}{C_0 \delta_0^2 T^2} \]
hold for some solution $\dot{x}_* = F_*(x_*)$ on some time interval $[0,T]$, then the following hold:

Any solution $y$ to the differential equation $\dot{y} = F(y)$ whose initial condition satisfies 
\[ \|y(0) - x_*(0)\| < \frac{1}{\sup_{t \in [0,T]} \|F_*(x_*(t))-F(x_*(t))\|} \]
in fact satisfies
\[ \|y(t) - x_*(t) \| < \frac{C_1}{\sup_{t \in [0,T]} \| F(x_*)-F_*(x_*)\|} \]
\end{lemma}

\begin{proof}
Introduce the new variable $Y$ by the equation $y = x_* + \eps Y$.  The proof will proceed by deriving first an evolution equation for $Y$, and then an evolution equation for $E:= 1+ \langle Y,Y\rangle$, using a bootstrapping argument to show that if $E(0)$ is not too large, then $E(t)$ remains not too large for $t \in [0,T]$.

We begin by computing
\[ \dot{Y} = \eps^{-1}(F(x_*+\eps Y) - F_*(x_*)) = JY + (F'(x_*)-J)Y + \eps \int_0^1\int_0^1 F''(x_* + \eps \rho_1\rho_2 Y)(Y,Y)d\rho_1d\rho_2 - \eps^{-1}(F_*(x_*) - F(x_*)) \]
The first term is the linearization about zero.  The second term is the linear part owing to the fact that the linearization about $x_*$ is not equal to the linearization about zero.  The third term incorporates all of the quadratic and higher order terms in $F$ and the fourth term owes to the fact that $x_*$ and $y$ satisfy different DEs.

In particular we have 
\begin{equation} \|\dot{Y}-JY\| \le \frac{\delta_0}{T}\|Y\| + \eps C_0 \|Y\|^2 + \eps^{-1}\| F(x_*)-F_*(x_*)\| \label{Y'-JY}\end{equation}

Define the almost conserved quantity $E = \frac{1}{2}\|Y\|^2 + 1$ and compute $\dot{E} = \langle \dot{Y}, Y \rangle = \langle \dot{Y}-JY,Y\rangle$.  

Thus $ |\dot{E}| \le \|\dot{Y}-JY\| \|Y\| \le \|\dot{Y}-JY\|(E+1)$ and hence
\[ E(t) \le E(0) e^{\sup_{t \in [0,T]} \|\dot{Y}-JY\| t} \]
for $0 \le t \le T$.
In light of \eqref{Y'-JY} we see that 
\[ E(t) \le 2 \mathrm{exp}\left( \delta_0 \sup_t \|Y(t)\| + C_0 T\eps \|Y\|^2  + T\eps^{-1} \| F(x_*)-F_*(x_*)\|\right) \]
Now let $\tau$ be the largest time for which $\sup_{t \in [0,\tau]} E(t) \le 2 \lceil e^4 \rceil$.  The hypotheses of the lemma guarantee that each of the terms $\delta_0 \|Y(t)\|$, $C_0 T\eps \|Y(t)\|^2$, $\delta_0 \|Y(t)\|$ and $T\eps^{-1}\| F(x_*)-F_*(x_*)\|$ are bounded above by one, hence the exponential of the sum is bounded above by $e^4$.  In particular $E(t) < E_0 e^4$ for as long as $E(t) < 2\lceil e^4 \rceil$ and also $t \in [0,T]$. Thus the inequality $E(t) < E_0 e^4$ holds for all $t \in [0,T]$.

\end{proof}

\begin{lemma}
Let $r$ denote a Toda two-soliton solution.  There is a constant $C$
such that $r_{n+1}-r_n < C k^3 e^{-C k n}$ with the constant $C$ uniform in $k$ 
and in $|t| < k^{-1}$.  
\end{lemma}

This follows
from a direct computation of the second and third differences
of the quantity $S_n$ given in Eq.~(\ref{toda_aux}).


\begin{thebibliography}{99}
\bibitem{Coste1997} C. Coste,  E. Falcon, and S. Fauve, 
Phys. Rev. E
{\bf 56}, 6104-6117 (1997).

\bibitem{Coste1999}	C. Coste,  and B. Gilles, 
 Eur. Phys. J. B  {\bf 7}, 155-168 (1999).

\bibitem{Coste2008}	C. Coste and B. Gilles, 
Phys. Rev. E 77, 021302 (2008).

\bibitem{Sen2005} R.S. Sinkovits and S. Sen, Phys. Rev. Lett. {\bf 74}, 2686 (1995); D.P. Visco, S. Swaminathan, T.R.Krishna Mohan, A. Sokolow and S. Sen, Phys. Rev. E {\bf 70}, 051306 (2004);  S. Sen, T.R.Krishna Mohan, D.P. Visco, S. Swaminathan, A. Sokolow, E. Avalos and M. Nakagawa, Int. J. Mod. Phys. B {\bf 19}, 2951 (2005).


\bibitem{Daraio2004}	C. Daraio,  V.F. Nesterenko, and S. Jin, 
AIP Conference Proceedings {\bf 706}, 197-200 (2004).

\bibitem{Daraio2006}	C. Daraio, and V.F. Nesterenko, 
Phys. Rev. E {\bf 73}, 026612 (2006).

\bibitem{deBilly2000}	M. de Billy, 
 J. Acoust. Soc. Amer. {\bf 108}, 1486-1495 (2000).

\bibitem{DeBilly2006}	M. De Billy, 
 Ultrasonics, {\bf 45}, 127-132 (2006).


\bibitem{Gilles2003} B. Gilles, and C. Coste, 
Phys. Rev. Lett.
{\bf 90}, 174302  (2003).

\bibitem{Nesterenko1983}	V.F. Nesterenko,  
 J. Appl. Mech. Tech. Phys.  {\bf 24}, 733-743 (1983).

\bibitem{Sen1998} S. Sen, M. Manciu and J.D. Wright, 
Phys. Rev. E {\bf 57}, 2386 (1998).

\bibitem{Nesterenko2001} V.F. Nesterenko,  Dynamics of Heterogeneous Materials, Springer-Verlag (New York, 2001).


\bibitem{Porter2008} M.A. Porter, C. Daraio, E.B. Herbold, I. Szelengowicz and P.G. Kevrekidis,
Phys. Rev. E {\bf 77},  015601 (2008).

\bibitem{Porter2009} M.A. Porter, C. Daraio, I. Szelengowicz, E.B. Herbold and P.G. Kevrekidis,
Physica D {\bf 238}, 666-676 (2009).

\bibitem{Rosas2004} A. Rosas,  and K. Lindenberg, 
Phys. Rev. E
{\bf 69}, 037601 (2004).

\bibitem{Sen2008}	S. Sen, J. Hong, J. Bang,  E. Avalosa, R. Doney,  
Phys. Rep. {\bf 462}, 21-66 (2008).

\bibitem{Shukla1991}	A. Shukla, 
 Optics and Lasers in Engineering, {\bf 14},  165-184 (1991).

\bibitem{Molinari2009} A. Molinari, and C. Daraio, 
Phys. Rev. E {\bf 80},  056602 (2009).

\bibitem{Herbold2007}	E.B. Herbold,  and V.F. Nesterenko, 
Shock 
Compression of Condensed Matter {\bf 955}, 231-234 (2007).


\bibitem{Herbold2007b}	E.B. Herbold, and V.F. Nesterenko, 
Appl. Phys. Lett. {\bf 90}, 261902 (2007).

\bibitem{Hertz1881}	H. Hertz, 
Journal fur die reine und angewandte Mathematik, {\bf 92}, 156-171 (1881).

\bibitem{Johnson1985} K.L. Johnson, Contact Mechanics,
Cambridge University Press (Cambridge 1985).

\bibitem{Sun2011} D. Sun, C. Daraio and S. Sen, Phys. Rev. E {\bf 83}, 066605 (2011).

\bibitem{Daraio2006b}	C. Daraio, V.F. Nesterenko, E.B. Herbold and S. Jin, 
Phys. Rev. Lett. {\bf 96}, 058002 (2006).

\bibitem{Hong2005} J. Hong,  
Phys. Rev. Lett. {\bf 94}, 108001
(2005).

\bibitem{Hong2002}	J.B. Hong, and A.G. Xu, 
Appl. Phys. Lett. {\bf 81}, 4868-4870 (2002).

\bibitem{Job2005} S.Job, F. Melo, A. Sokolow and S. Sen,
Physical Review Letters {\bf 94}, 178002 (2005).

\bibitem{Melo2006}	F. Melo, S. Job,
F. Santibanez, and F. Tapia, 
Phys. Rev. E {\bf 73}, 041305 (2006).




\bibitem{Nesterenko2005}	V.F. Nesterenko, C. Daraio, E.B. Herbold
and S. Jin, 
Phys. Rev. Lett. {\bf 95}, 158702
(2005).

\bibitem{Boechler2010}	N. Boechler, G. Theocharis, S. Job,
P.G. Kevrekidis, M.A. Porter and C. Daraio, 
Phys. Rev. Lett. {\bf 104}, 244302 (2010).


\bibitem{Job2009}	S. Job, F. Santibanez, F. Tapia
and F. Melo, 
Phys. Rev. E {\bf 80}, 025602(R) (2009).

\bibitem{Theocharis2009} G. Theocharis, M. Kavousanakis, P.G. Kevrekidis,
C. Daraio, M.A. Porter and I.G. Kevrekidis,
Phys. Rev. E {\bf 80}, 066601 (2009).

\bibitem{chong2013} C. Chong, F. Li, J. Yang, M.O. Williams, I.G. Kevrekidis, P. G. Kevrekidis, C. Daraio,
 Phys. Rev. E {\bf 89}, 032924 (2014).

\bibitem{sen1} S. Job, F. Melo, A. Sokolow, and S. Sen
Phys. Rev. Lett. {\bf 94}, 178002 (2005).

\bibitem{sen2} R. Doney and S. Sen
Phys. Rev. Lett. {\bf 97}, 155502 (2006).

\bibitem{sen3} S. Sen and T.R.Krishna Mohan
Phys. Rev. E {\bf 79}, 036603 (2009).

\bibitem{sen4} E. {\'A}valos and S. Sen,
Phys. Rev. E {\bf 79}, 046607 (2009).

\bibitem{vakakis1} Y. Starosvetsky and A.F. Vakakis,
Phys. Rev. E {\bf 82}, 026603 (2010).

\bibitem{vakakis2} K. R. Jayaprakash, Y. Starosvetsky, and A.F. Vakakis
Phys. Rev. E {\bf 83}, 036606 (2011).

\bibitem{vakakis3} I. Szelengowicz, M. A. Hasan, Y. Starosvetsky, A. Vakakis, and C. Daraio
Phys. Rev. E {\bf 87}, 032204 (2013).

\bibitem{pegoenglish} J. M. English and R. L. Pego, Proceedings of the AMS 
{\bf 133}, 1763 (2005).

\bibitem{ahnert} K. Ahnert and A. Pikovsky
Phys. Rev. E {\bf 79}, 026209  (2009).

\bibitem{stefanov1} A. Stefanov and P.G. Kevrekidis, J. Nonlin. Sci. 
{\bf 22}, 327 (2012).

\bibitem{stefanov2} A. Stefanov and P.G. Kevrekidis, 
Nonlinearity {\bf 26}, 539 (2013).

\bibitem{bookfpu} G. Gallavotti, 
{\it The Fermi-Pasta-Ulam Problem: A Status Report},
Springer-Verlag (New York, 2010).

\bibitem{Kruskal1963} N. J. Zabusky and M. D. Kruskal,
Phys. Rev. Lett. {\bf 15}, 240 (1965).

\bibitem{Schneider1999} G. Schneider and C.E. Wayne, 
{\it Counter-progagating waves on fluid surfaces and the continuum
limit of the Fermi-Pasta-Ulam model}, In K. Fiedler, B. Gr{\"o}ger and 
J. Sprekels (Eds.),
EQUADIFF’99; Proceedings of the International Conference on Differential Equations, World Scientific,
(Singapore, 2000).




\bibitem{pego1} G. Friesecke and R.L. Pego, Nonlinearity 
{\bf 12}, 1601 (1999).

\bibitem{pego2} G. Friesecke and R.L. Pego, Nonlinearity 
{\bf 15}, 1343 (2002).

\bibitem{pego3} G. Friesecke and R.L. Pego, Nonlinearity 
{\bf 17}, 207 (2004).

\bibitem{pego4} G. Friesecke and R.L. Pego, Nonlinearity 
{\bf 17}, 2229 (2004).

\bibitem{hoffman} A. Hoffman and C.E. Wayne,
Nonlinearity {\bf 21} 2911 (2008).

\bibitem{dcds}  C. Chong, P.G. Kevrekidis and G. Schneider,
Discr. Cont. Dyn. Sys. A {\bf 34}, 3403 (2014).

\bibitem{jkyang} F. Li, L. Zhao, Z. Tian, L. Yu, J. Yang, 
Smart Materials and Structures
{\bf 22}, 035016 (2013).

\bibitem{hir} R. Hirota, 
Phys. Rev. Lett. {\bf 27}, 1192 (1971).

\bibitem{wahl} H.D. Wahlquist and F.B. Estabrook,
Phys. Rev. Lett. {\bf 31}, 1386 (1973).

\bibitem{Marchant} T. R. Marchant, Phys. Rev. E 59, 3745 (1999).

\bibitem{toda1}
M. Toda, Prog. Theor. Phys. Suppl. No. 45, 174 (1970).

\bibitem{toda2}
M. Toda, Theory of Nonlinear Lattices, Springer, New York, 1981.

\bibitem{toda3}
M. Toda, M. Wadati, J. Phys. Soc. Jpn.  {\bf 34}, 1 (1973).

\bibitem{sen00} M. Manciu, S. Sen and A.J. Hurd, Phys. Rev. {\bf 63}, 016614 (2000); F.S. Manciu and S. Sen, Phys. Rev. E {\bf 66}, 016616 (2002).

\bibitem{jobmelo} F. Santibanez, R. Munoz, A. Caussarieu, S. Job and
F. Melo, Phys. Rev. E {\bf 84}, 026604 (2011).


\end{thebibliography}
\end{document}